%% file: root.tex
\documentclass[11pt,a4paper,twoside,reqno]{amsart}
\input{style}
\input{commands}

\usepackage{epsfig} % for postscript graphics files
\usepackage{graphicx}
\usepackage{grffile}
\usepackage{amsmath}
\usepackage{enumitem}
\usepackage{booktabs}
\usepackage{multirow}
\usepackage{array}
\usepackage{cite}
\usepackage{bm}
\usepackage{pdfpages} 
\usepackage{subfigure}
\usepackage{booktabs}
\usepackage{xcolor}
\usepackage{algorithm}
\usepackage{algpseudocode}
\usepackage{wrapfig}
\usepackage{caption}

\begin{document}

\title[Multivariate FDE over finite-frequency range]{Robust Multivariate Detection and Estimation with Fault Frequency Content Information}
\author{Jingwei Dong, Kaikai Pan, S\'ergio Pequito and Peyman Mohajerin Esfahani}
\thanks{The authors are with the Delft Center for Systems and Control, Delft University of Technology, The Netherlands (\{J.Dong-6, P.MohajerinEsfahani\}@tudelft.nl), the College of Electrical Engineering, Zhejiang University, China (pankaikai@zju.edu.cn), and the Department of Electrical and Computer Engineering, University of Lisbon, Lisbon (sergio.pequito@tecnico.ulisboa.pt). This work is partially supported by the ERC grant TRUST-949796, CSC (China Scholarship Council) with funding number 201806120015, the National Natural Science Foundation of China under Grant 62201501 and Grant 52161135201.}
\maketitle

\begin{abstract}
    This paper studies the problem of fault detection and estimation (FDE) for linear time-invariant (LTI) systems with a particular focus on frequency content information of faults, possibly as multiple disjoint continuum ranges, and under both disturbances and stochastic noise.
    To ensure the worst-case fault sensitivity in the considered frequency ranges and mitigate the effects of disturbances and noise, an optimization framework incorporating a mixed $\Hmin/\Hto$ performance index is developed to compute the optimal detection filter. 
    Moreover, a thresholding rule is proposed to guarantee both the false alarm rate (FAR) and the fault detection rate~(FDR).
    Next, shifting attention to fault estimation in specific frequency ranges, an exact reformulation of the optimal estimation filter design using the restricted~$\Hinf$ performance index is derived, which is inherently non-convex. However, focusing on finite frequency samples and fixed poles, a lower bound is established via a highly tractable quadratic programming (QP) problem.
    This lower bound together with an alternating optimization (AO) approach to the original estimation problem leads to a suboptimality gap for the overall estimation filter design.
    The effectiveness of the proposed approaches is validated through applications of a non-minimum phase hydraulic turbine system and a multi-area power system.  
\end{abstract}

\section{Introduction}
Fault diagnosis has been the focus of research in the past decades due to its critical importance in ensuring the safety and reliability of various engineering systems, such as power networks, vehicle dynamics, and aircraft systems~\cite{hwang2009survey,gao2015survey}.
Timely and accurate FDE of faults while a system is still operating in a controllable condition, can help prevent further damage and reduce losses.
However, FDE performance is inevitably affected in practice by model uncertainties, disturbances, and stochastic noise, which can result in false alarms, missing detection, and large estimation errors. Hence, it is essential to consider these interferences when designing FDE methods.

In recent years, there also has been growing recognition of the need to address faults in specific frequency ranges.
This stems from the fact that many practical faults (or cyber-attack signals~\cite{pan2019static}) exhibit distinct frequency characteristics, e.g., incipient faults in low-frequency ranges and actuator stuck faults with zero frequency~\cite{wang2008finite}. 
Existing FDE methods developed for the entire frequency range can cause conservatism when dealing with these faults. 
Motivated by the above issues, this study focuses on the FDE problem in specific frequency ranges, considering both disturbances and stochastic noise.

\textbf{Fault detection:} 
A number of model-based fault detection methods have been developed for dynamical systems with disturbances and noise. 
The basic idea is to design residual generators using observer-based or parity-space approaches~\cite{gao2015survey}.
The outputs of residual generators (called residuals), that are used to indicate the occurrence of faults, should be sensitive to faults and robust to disturbances and noise, simultaneously.  
To this end, performance indices, such as~$\Hinf$ and~$\Hto$ norms are employed to measure the robustness against disturbances and noise. 
The~$\Hmin$ index, representing the worst-case fault sensitivity, is incorporated into the design of residual generators.
For instance, the authors in~\cite{hou1996lmi} first proposed the~$\Hmin/\Hinf$ observer. 
Another residual generation method~\cite{nyberg2006residual} developed in the framework of differential-algebraic equations (DAE) has attracted attention these years. 
This method can find residual generators of the possibly lowest order compared to conventional \mbox{observer-based} or parity-space approaches. 
Moreover, it offers much design freedom due to the ability to characterize all possible residual generators for systems represented by DAE.
As a result, different fault detection methods have been developed in the DAE framework, such as accounting for nonlinear terms~\cite{esfahani2015tractable} and modeling uncertainties~\cite{pan2021dynamic}.

Note that the above methods all consider the entire frequency range, where conservatism exists and the~$\Hmin$ index will be zero for strictly proper systems.
The authors in~\cite{liu2005lmi} addressed this issue by introducing a weighting function to enhance the $\Hmin$ index in a specific frequency range, and further provided the existing condition of a non-zero~$\Hmin$ index. 
However, finding an appropriate weighting function is complex.
In contrast, the generalized Kalman-Yakubovich-Popov (GKYP) lemma in~\cite{iwasaki2005generalized} provides a way to directly constrain the~$\Hmin$ index in a frequency range.
Based on the GKYP lemma, the authors in~\cite{wang2008finite} employed the~$\Hmin/\Hinf$ index to design a Luenberger observer for fault detection of LTI systems with enhanced fault sensitivity in a specific frequency range.
Furthermore, the integration of~$\Hmin/\Hinf$ index and the GKYP lemma has been incorporated into the design of fault detection approaches for linear parameter-varying descriptor systems~\cite{wang2017h} and nonlinear systems~\cite{li2020fault}.
Considering that the~$\ell_{\infty}$ norm representing the peak value of a signal is more suitable for residual evaluation compared to the $\Hinf$ norm, the authors in~\cite{tang2020fault} chose the~$\Hmin / \ell_{\infty}$ index to design the fault detection observer for linear descriptor systems.
For a more comprehensive analysis of different indices used in fault detection problems, such as~$\Hmin/\Hinf$, $\Hinf/\Hinf$, and~$\Hto/\Hinf$, see~\cite{liu2007optimal}.

In addition, the~$\Hmin$ index has been investigated in the time domain as well, where fault sensitivity in a finite or infinite time horizon is maximized, see for example~\cite{khan2014fault,han2022fault}.
It is worth mentioning that the aforementioned methods using the~$\Hinf$ and~$\ell_{\infty}$ norms typically consider disturbances or noise with bounded energy or peak values, which results in conservative diagnosis results. Moreover, the deterministic bounds are generally difficult to obtain in practical scenarios~\cite{boem2018plug}. 
Therefore, exploiting the stochastic nature of these signals can be a promising alternative.
Moreover, to our knowledge, little attention has been paid to designing residual generators for fault detection within specific frequency ranges, accounting for both disturbances and stochastic noise.

\textbf{Fault estimation:} 
Accurate fault estimation that provides the size and shape of faults is a fundamental task in the fault diagnosis area.
Many model-based fault estimation methods are based on observers~\cite{gao2015unknown,ghanipoor2023robust}, which generally require fault signals to be finitely differentiable.
Different from \mbox{observer-based} methods, fault estimation filters do not require estimates of system states and assumptions regarding the derivatives of fault signals, such as the \mbox{system-inversion-based} fault estimation filters developed in~\cite{wan2017fault}.
However, the existence of a stable \mbox{system-inversion-based} estimation filter cannot be ensured when there are unstable zeros (i.e., in non-minimum-phase systems).
Another approach to designing fault estimation filters is directly minimizing the difference between the transfer function of the fault subsystem and the identity matrix in the~$\Hinf$ optimization framework, as presented in~\cite{niemann2000design}.
Once again, the above estimation methods are for the entire frequency range. 

The existing methods for fault estimation in the frequency domain are primarily built on observer-based methods and the GKYP lemma.
The authors in~\cite{wei2010robust} designed a fault estimation observer for LTI systems, where the $\Hinf$ norm defined in a specific frequency range was employed to mitigate the effects of disturbances and faults on estimation errors.
The result was then extended and applied to Takagi–Sugeno fuzzy systems~\cite{zhang2014analysis} and descriptor systems~\cite{wang2019robust}.
However, the design of fault estimation filters considering fault frequency content information has received considerably less attention.
To the best of our knowledge, only~\cite{ding2008model} and~\cite{stefanovski2019input} investigated this problem. 
In particular, the authors in~\cite[Theorem 14.6]{ding2008model} incorporated a weighting function into the~$\Hinf$ optimization framework to improve fault estimation performance in a specific frequency range. However, as mentioned before, the selection process of a proper weighting function is complex.
The recent result~\cite{stefanovski2019input} designed the fault estimation filter represented by a rational matrix with constant inertia in the frequency region to attenuate disturbances, but it only considered fault estimation in the steady-state.
Therefore, developing a tractable design method for fault estimation filters in the frequency domain capable of dealing with disturbances, stochastic noise, and a broader class of faults is meaningful.

\textbf{Main contributions:}
In view of the existing results mentioned above, this study pioneers the design of FDE filters exploiting fault frequency content information in the DAE framework.
Compared to the existing results focusing on FDE in the frequency domain, the proposed design framework offers the following key features: 
(i) it can deal with disturbances and stochastic noise and does not require assumptions on the derivatives of fault signals, thus applicable to a larger class of fault diagnosis problems; 
(ii) it produces FDE filters of the possibly lowest order compared to observer-based methods; 
(iii) it offers design flexibility by allowing for residuals of arbitrary dimensions and enabling the simultaneous design of both the numerator and denominator of FDE filters, while other fault diagnosis methods developed within the DAE framework typically design one-dimensional residuals with fixed denominators~\cite{esfahani2015tractable,pan2021dynamic}; 
(iv) the design of FDE filters, which considers fault frequency content spanning multiple disjoint continuum ranges, is formulated into a unified optimization framework using the GKYP lemma. This approach significantly simplifies the design process of FDE filters in the frequency domain. 
Note that the derived optimization problems for filter design are inherently non-convex, for which an efficient approach is developed to approximate a suboptimal solution along with explicit performance bounds.
The contributions of this paper are summarized as follows: 
\begin{itemize}
    \item {\bf Optimal detection with fault frequency content:} 
    The design of the fault detection filter, utilizing~$\mathcal{H}_{\_} / \mathcal{H}_2$ index in the DAE framework, is formulated as a finite optimization problem (Theorem~\ref{thm:FD filter design}). This enables the derived filter to handle disturbances and stochastic noise while enhancing fault sensitivity across the set of disjoint continuum frequency ranges.  

    \item {\bf Thresholding with false alarm rate and fault detection rate guarantees:} 
    A thresholding rule that provides guarantees on FAR and FDR (Theorem~\ref{Thm:performance certificates}) is developed, which improves the current literature (e.g.,~\cite{boem2018plug,dong2023multimode}) by extending the setting to multivariate residuals and ensuring FAR and FDR simultaneously.
     
    \item{\bf Optimal estimation with fault frequency content:} 
    Shifting attention from detection to estimation, the $\mathcal{H}_{\_}$ index is replaced with the "restricted"~$\mathcal{H}_{\infty}$ norm in specific frequency ranges. 
    The fault estimation filter design is then reformulated in the DAE framework as a finite optimization problem (Theorem~\ref{Thm: exact FE}).
    In contrast to the existing estimation results that focus on faults represented by either step signals~\cite{van2022multiple} or polynomials~\cite{ghanipoor2023robust}, this study considers a larger class of faults with frequency content containing multiple disjoint continuum ranges.
 
    \item {\bf Convex approximation with suboptimality gap:} 
     By relaxing frequency ranges to finitely many samples, the estimation problem is lower bounded by a QP problem (Theorem~\ref{Thm:FE filter design}), whose solution can be approximated by a \mbox{closed-form} formula~(Corollary~\ref{cor:Analytical solution}). 
    Combining this with an AO approach to the original estimation problem yields a suboptimality gap for the overall design with given fixed filter poles (Proposition~\ref{prop: opt gap}). 
\end{itemize}

The rest of the paper is organized as follows. 
The problem formulation is introduced in Section~\ref{sec:problem description}. Section~\ref{sec:FD} presents design methods for the fault detection filter and the thresholding rule. 
In Section~\ref{sec:FE}, design methods for the fault estimation filter and the derivation of the suboptimality gap are developed.
To improve the flow of the paper and its accessibility, some technical proofs are relegated to Section~\ref{sec:proofs}.
The proposed approaches are applied to a non-minimum phase system and a multi-area power system in Section~\ref{sec:simulation} to demonstrate their effectiveness.
Finally, Section~\ref{sec:conslusion} concludes the paper with future directions.  

\paragraph{\bf Notation} Sets~$\N$,~$\R~(\R_+)$, and~$\R^n$ denote \mbox{non-negative} integers, (positive) reals, and the space of~$n$ dimensional \mbox{real-valued} vectors, respectively. 
The set of symmetric and Hermitian matrices are denoted by~$\mathcal{S}$ and~$\mathbb{H}$, respectively.
The identity matrix with an appropriate dimension is denoted by~$I$. 
For a vector~$v = [v_1,\dots,v_{n}]^{\top}$, the $\infty$-norm and $2$-norm of~$v$ are~$\|v\|_{\infty} = \max_{i\in\{1,\dots,n\}} |v_i|$ and~$\|v\|_2 = \sqrt{\sum^n_{i=1} v^2_i}$, respectively. 
For a matrix~$A$, the 2-norm and Frobenius norm are denoted by~$\|A\|_2$ and~$\|A\|_F$, respectively.
For a random variable~$\chi$, the probability law and its expectation are denoted by~$\mPr[\chi]$ and~$\E[\chi]$, respectively.
Given a discrete-time signal~$u = \{u(k)\}_{k\in\N}$ and a transfer function~$\T$, the notation~$\T[u]$ denotes the output in response to~$u$. 
The~$\Lto$-norm of~$u$ is~$\|u\|^2_{\Lto} = \sum^{\infty}_{k=0} u^{\top}(k)u(k)$.
With a slight abuse of notation, $*$ is used to denote the off-diagonal elements in symmetric (or Hermitian) matrices to avoid clutter, and~$A^*$ to denote the complex conjugate transpose of the matrix~$A$. 
The transpose of~$A$ is denoted by~$A^{\top}$. A positive definite (semi-definite) matrix is denoted by~$A \succ 0 (\succeq 0)$.

%%%%%%%%%%%%%%%%%%%%%%%%%%%%%%%%%%%%%%%%%%%%%%%%%%%%%%%%%%%%%%%%%%%%%%%%%%%
%%%%%%%%%%%%%%%%%%%%%%%%%%%%%%%%%%%%%%%%%%%%%%%%%%%%%%%%%%%%%%%%%%%%%%%%%%%
%                        Problem formulation
%%%%%%%%%%%%%%%%%%%%%%%%%%%%%%%%%%%%%%%%%%%%%%%%%%%%%%%%%%%%%%%%%%%%%%%%%%%
\section{Model Description and Problem Statement}\label{sec:problem description}
Consider the following discrete-time LTI system 
\begin{equation}\label{eq:SS model}
    \left\{ \begin{array}l
         x(k+1) = Ax(k) +Bu(k) +B_d d(k) + B_{\omega} \omega(k) +B_f f(k)  \\
         y(k) = Cx(k) + Du(k) + D_{\omega}\omega(k) + D_f f(k),
    \end{array}
    \right.
\end{equation}
where~$x(k) \in \R^{n_x}$,~$u(k) \in \R^{n_u}$,~$d(k) \in \R^{n_d}$, and~$y(k) \in \R^{n_y}$ are the state, control input, disturbance, and measurement output, respectively.
The signal~$\omega(k) \in \R^{n_{\omega}}$ denotes the independent and identically distributed (i.i.d.) white noise with zero mean. 
The signal~$f(k) \in \R^{n_f}$ denotes the fault.
System matrices in~\eqref{eq:SS model} are all known with appropriate dimensions. Throughout this study, our filter design is restricted to a subclass of fault signals with the following frequency content. 

\begin{As}[Fault frequency content]\label{As:fault regularity}
        The fault signal frequency content, also referred to as the signal spectrum, is the union of the disjoint intervals~$\bar{\Theta} := \cup_{m \in \{1,\dots,n_{\theta}\}} \Theta_m$ where $\Theta_{m} \subset [-\pi, \pi]$ and~$\Theta_{m_1} \cap \Theta_{m_2} = \varnothing$ for all~$m_1\neq m_2$. 
        In other words, the fault signal can be fully characterized in the frequency domain via $f(t) = \int_{\bar{\Theta}} F(\e^{j\theta}) \e^{j\theta t} \diff \theta$ where~$F(\e^{j\theta})$ is the \mbox{Discrete-Time} Fourier Transform.  
        This class of fault signals is denoted by~$\mathcal{F}(\bar{\Theta})$.
\end{As}

The objective of this work is to design filters that can detect and estimate faults with frequency content~$\bar{\Theta}$ through the control input $u$ and the measurement $y$.
To this end, we consider filters in the DAE framework and introduce the time-shift operator~$\mfq$, i.e.,~$x(k+1)=\mfq x(k)$. 
Then, the state-space model~\eqref{eq:SS model} is transformed into the DAE format
\begin{equation}\label{eq:DAE sys}
     H(\mfq)\begin{bmatrix}x \\d \end{bmatrix} 
     + L \begin{bmatrix} y \\ u \end{bmatrix} 
     + W [\omega]  + G [f] + \begin{bmatrix} x_0 \\ 0\end{bmatrix} = 0,
\end{equation}
where~$x(0) = x_0$ is the unknown initial condition, the polynomial matrices~$H(\mfq)$,~$L$,~$W$ and~$G$ are given by
\begin{align*}
&H(\mfq) =  H_1 \mfq +H_0 
=\begin{bmatrix}
	-\mfq I+A    &B_d\\
	C                &0
\end{bmatrix}, 
    ~H_0 = 
    \begin{bmatrix}
        A &B_d \\ C &0   
    \end{bmatrix}, \\
    &H_1 = 
    \begin{bmatrix}
        -I &0 \\ 0 &0   
    \end{bmatrix}, 
    ~L=\begin{bmatrix} 0 &B\\ -I &D \end{bmatrix},
	~W = \begin{bmatrix} B_{\omega} \\ D_{\omega} \end{bmatrix},
	~\text{and}~G = \begin{bmatrix} B_f \\D_f \end{bmatrix}.
\end{align*} 
Given the DAE format of the system, the filter is defined as
\begin{equation}\label{eq: FD filter}
r = \F (\mfq)L \begin{bmatrix}  y\\u \end{bmatrix}, \quad \F(\mfq) := \frac{\mathcal{N}(\mfq)}{a(\mfq)},
\end{equation}
where $r\in \R^{n_r}$ is the residual, $\mathcal{N}(\mfq)= \sum^{d_N}_{i=0} N_i \mfq^i$ is a polynomial matrix with coefficients $N_i \in \R^{n_r \times (n_x +n_y)}$ and degree~$d_N$. 
The denominator is~$a(\mfq) = \sum^{d_a}_{i=0} a_i \mfq^i + \mfq^{d_a+1}$, where~$a_i\in\R$ and~$d_a+1$ is the degree of~$a(\mfq)$ with~$d_a \geq d_N$ to ensure that the filter is strictly proper. Note that the parameters of $\F(\mfq)$, i.e.,~$N_i$ and~$a_i$, are the filter variables to be determined.

Multiplying~\eqref{eq:DAE sys} from the left side by~$\F(\mfq)$ , the residual~$r$ becomes
\begin{align}\label{eq:residual format}
    r = \F(\mfq) L \begin{bmatrix} y \\ u \end{bmatrix} 
      =-\F(\mfq) H(\mfq) [X]  
       - \F(\mfq)W [\omega]  - \F(\mfq)G [f] 
       - \F(\mfq) \begin{bmatrix} x_0 \\0\end{bmatrix}, 
\end{align}
where~$X = [x^{\top} ~d^{\top}]^{\top}$. 
The right-hand side of~\eqref{eq:residual format} indicates the input-output relations from~$X,~\omega$, and~$f$ to~$r$, based on which one can design~$\F(\mfq)$ such that desired mapping relations are satisfied for different diagnosis purposes.
Subsequently, for the sake of exposition, these mapping relations are denoted as
\begin{align*}
    ~\T_{Xr}(\mfq)=-\F(\mfq)H(\mfq),
    ~\T_{\omega r}(\mfq)=-\F(\mfq)W, 
    ~\T_{fr}(\mfq)=-\F(\mfq)G.
\end{align*}

\begin{As}[Initial condition dependency]\label{As: initial con}
    The contribution of the initial condition, i.e., the last term in~\eqref{eq:residual format}, vanishes exponentially fast under appropriate stability conditions.
\end{As}
Assumption~\ref{As: initial con} is commonly adopted in fault detection literature~\cite{wan2016data,shang2021distributionally}. 
Next, the two problems studied in this work are presented, including (i) fault detection (Section~\ref{ss:FD}), and (ii) fault estimation (Section~\ref{ss:FE}).

\subsection{Problem 1: Fault detection} \label{ss:FD}
In order to formally introduce the fault detection problem, the~$\Hto$ norm and $\Hmin$ index of a transfer function, e.g.,~$y = \T(\mfq) [u], ~\T(\mfq)=\mathcal{C}(\mfq I-\mathcal{A})^{-1}\mathcal{B}$, are introduced as follows.

\begin{Def}[$\Hto$ norm~\cite{scherer1997multiobjective}] \label{def:Hto}
Assume~$\mathcal{A}$ is stable. The $\Hto$ norm of~$\T(\mfq)$ is defined as
\begin{align*}
    \left\|\T (\mfq) \right\|^2_{\Hto} = \frac{1}{2\pi} \int^{\pi}_{-\pi} \Tr \left(\T^{*}(\e^{j \theta})\T(\e^{j \theta})\right) \textup{d} \theta ,
\end{align*}
and corresponds to the asymptotic variance of the output when the system is driven by the white noise with zero mean.
\end{Def}

\begin{Def}[{$\Hmin$ index~{\cite{liu2005lmi}}}]\label{def:Hmin}
The $\Hmin$ index of~$\T(\mfq)$ in a single continuum frequency range~$\Theta$ is defined as
\begin{align*}
    \left\|\T (\mfq) \right\|_{\Hmin(\Theta)} = \inf\limits_{\theta \in \Theta, u \neq 0} \frac{\left\|\T(\e^{j \theta}) u \right\|_{\Lto}}{\|u\|_{\Lto}},
\end{align*}
which can also be rewritten as~$\left\|\T (\mfq) \right\|_{\Hmin(\Theta)} = \inf\limits_{\theta \in \Theta} \underline{\sigma} \left(\T (\e^{j \theta})\right)$ with~$\underline{\sigma}(\cdot)$ denoting the minimum singular value. 
\end{Def}

Let us look into the right-hand side of~\eqref{eq:residual format}. For fault detection problem, the residual~$r$ is expected to be insensitive to~$d$, robust to~$\omega$, and sensitive to~$f$ in~$\mathcal{F}(\bar{\Theta})$. First, to decouple $d$ from $r$, it needs to guarantee that
\begin{subequations}\label{eq:mapping relations}
    \begin{align}
    \T_{Xr}(\mfq)=-\F(\mfq)H(\mfq) = 0.  \label{eq:mapping 1}
\end{align}
Second, an upper bound~$\eta_1 \in \R_+$ is set on the~$\Hto$ norm of~$\T_{\omega r}(\mfq)$, to suppress the contribution of~$\omega$ to $r$, as
\begin{align}
    \|\T_{\omega r}(\mfq)\|^2_{\Hto} = \|-\F(\mfq)W\|^2_{\Hto} \leq \eta_1, \label{eq:mapping 2}
\end{align}
which also ensures the stability of the filter based on the classical result of $\Hto$ norm. Finally, the~$\Hmin$ index of~$\T_{fr}(\mfq)$ in~$\bar{\Theta}$ is enforced to be larger than some positive value~$\eta_2\in \R_+$ to guarantee the worst-case fault sensitivity, i.e.,  
\begin{align}
    \|\T_{fr}(\mfq)\|^2_{\Hmin(\Theta_m)}= \|-\F(\mfq)G\|^2_{\Hmin(\Theta_m)} \geq \eta_2, ~\forall \Theta_m \subset \bar{\Theta}.  \label{eq:mapping 3}
\end{align}
\end{subequations}

In view of the desired mapping conditions~\eqref{eq:mapping relations}, the design of the fault detection filter is formulated as the following optimization problem.

\paragraph{{\bf Problem 1a} \textup{(Fault detection filter design)}\bf}
\textit{
Consider the system~\eqref{eq:SS model}, the filter to be designed in~\eqref{eq: FD filter}, and the expression of the residual~\eqref{eq:residual format}. Given a scalar~$\alpha \in [0,1]$, find~$\F(\mfq)$ via the minimization program:
\begin{align*}
    \min_{\eta_1,\eta_2 \in \R_+, ~\F(\mfq)} \quad \{\alpha \eta_1-(1-\alpha)\eta_2:~\eqref{eq:mapping 1},~\eqref{eq:mapping 2},~\eqref{eq:mapping 3}\}.
\end{align*}}

The following assumption is introduced to guarantee the feasibility of Problem~1a.
\begin{As}[Feasibility condition]\label{As: feasibility}
    The pair~$(A,C)$ is observable. For~$\mfq = \varphi \e^{j\theta}$ with~$|\varphi|>1$ and~$\theta \in \bar{\Theta}$, the following rank condition holds
    \begin{small}
    \begin{align*}
        n_x+n_y \geq \textup{Rank}\left(\begin{bmatrix}
        -\mfq I + A  &B_d &B_f \\
        C              &0 &D_f
        \end{bmatrix}\right) 
        = n_x + \textup{Rank} \left(\begin{bmatrix}
        B_d\\ 0 
        \end{bmatrix}\right) + n_f.
    \end{align*}
    \end{small}
\end{As}
Denote the transfer functions from $d$ to $y$ and $f$ to $y$ by~$\T_{dy}(\mfq) = C(\mfq I-A)^{-1}B_d$ and~$\T_{f y}(\mfq) = C(\mfq I-A)^{-1}B_f + D_f$, respectively. 
It readily follows 
\begin{align*}
    n_y \geq \textup{Rank}[\T_{dy}(\mfq) \quad\T_{fy}(\mfq)] = \textup{Rank}\left(\begin{bmatrix}
    B_d \\0
\end{bmatrix}\right) + n_f,
\end{align*}
if Assumption~\ref{As: feasibility} holds \cite[Theorem 6.2]{ding2008model}. 
Therefore, Assumption~\ref{As: feasibility} ensures simultaneously the following: (i) the disturbance~$d$ can be decoupled, and (ii) the fault~$f$ satisfies input observability condition in~$\bar{\Theta}$, which also indicates that there are no unstable invariant zeros in~$\bar{\Theta}$. 
The second term is necessary for a nonzero~$\Hmin$ index~\cite[Lemma 5]{liu2005lmi}. 
Note that the fault frequency content information is incorporated into the analysis, which is derived from the classical result on the input observability condition in~\cite[Theorem 3]{hou1998input} and~\cite[Corollary 14.1]{ding2008model}.

Additionally, a solution to Problem 1a ensures that the residual $r$ can be written as
\begin{align*}
    r = \T_{\omega r}(\mfq) [\omega] + \T_{f r}(\mfq)[f],
\end{align*}
where no dependency on~$X$ is present because it is decoupled. 
In practice, the residual $r$ will oscillate around zero as a response to the noise~$\omega$ in the absence of~$f$. In contrast, the residual will ideally be away from zero when a fault happens.
Subsequently, let us take the average $2$-norm of~$r$ over a time interval~$\mathcal{T} \in \N$ as the evaluation function, i.e.,
\begin{align}\label{eq: Jr}
    J(r) = \frac{1}{\mathcal{T}}  \sum^{k_1+\mathcal{T}}_{k=k_1} \|r(k)\|_2,
\end{align}
where $k_1\in\N$. Given a threshold~$J_{th} \in \R_+$, the following fault detection logic is introduced:
\begin{align*}
    \left\{\begin{array}{ll}
        J(r) \leq J_{th} &\Rightarrow \quad\textup{no fault alarm},  \\
        J(r) > J_{th}    &\Rightarrow \quad\textup{fault alarm}.
    \end{array} \right.
\end{align*}
Note that false alarms and missing detection of faults are inevitable due to the random nature of noise. 
To tackle these issues, a threshold $J_{th}$ that can provide guarantees on FAR and FDR is considered in the following problem.

\paragraph{{\bf Problem 1b} \textup{(Thresholding with guarantees on FAR and FDR)}\bf} 
\textit{
Given the fault detection filter constructed from Problem 1a, an acceptable FAR~$\varepsilon_1 \in (0,1]$, and a set of fault signals of interest $\Omega_f:=\{f: \|f(k)\|_2 \geq \underline{f},~\underline{f} \in \R_+,~f \in \mathcal{F}(\bar{\Theta})\}$, determine the threshold~$J_{th}$ such that:
\begin{subequations}
    \begin{align}
    &\textup{FAR:} ~\mPr\left\{J(r) > J_{th} \big| f = 0 \right\} \leq \varepsilon_1, \label{eq:FAR prob} \\
    &\textup{FDR:} ~\mPr\left\{J(r) > J_{th} \big| f\in\Omega_f \right\} \geq \varepsilon_2,\label{eq:FDR prob}
\end{align}
\end{subequations}
where $\varepsilon_2$ is the lower bound on FDR to be computed.}

\begin{Rem}[Difficulty in FDR computation]
    There are fewer results in the literature on FDR computation because different elements of multivariate fault signals may cancel out each other’s contributions to the residual~\cite{pan2019static}. 
    As a result, there is no guarantee that FDR even exists.
    By assuming that a set of faults is detectable, authors in~\cite[Section 12.1]{ding2008model} propose a computation method of FDR in the norm-based framework.
    In this work, the~$\Hmin(\Theta)$ index is employed to ensure fault sensitivity, which paves the path for FDR computation in a stochastic way. 
\end{Rem}

\subsection{Problem 2: Fault estimation} \label{ss:FE}
In certain scenarios, it becomes essential not just to identify the occurrence of faults, but also to estimate them precisely.
For instance, incorporating fault estimates into fault-tolerant controllers is a common practice to counteract the effects of faults~\cite{gao2015fault}. 
Here, to ensure that the residual follows fault signals within~$ \mathcal{F}(\bar{\Theta})$, a stable~$\F(\mfq)$ in~\eqref{eq: FD filter} is determined such that the subsequent relation holds
\begin{align}\label{eq: estimation condition}
    \frac{\| \T_{fr}(\mfq) f - f\|^2_{\Lto}}{\|f\|^2_{\Lto}} \leq \eta_3, \quad \forall ~f \in \mathcal{F}(\bar{\Theta}),
\end{align}
where $\eta_3 \in \R_+$ is an upper bound. 
The estimation condition~\eqref{eq: estimation condition} is consistent with the format of the restricted~$\Hinf$ norm in a specific frequency range. 

\begin{Def}[Restricted $\Hinf$ norm~\cite{gao2011h}]\label{def:Hinf}
The restricted $\Hinf$ norm of a transfer function~$\T(\mfq)$ in a single continuum frequency range~$\Theta$ is defined as
\begin{align*}
    \left\|\T (\mfq) \right\|_{\Hinf(\Theta)} = \sup\limits_{\theta \in \Theta, u \neq 0} \frac{\left\|\T(\e^{j \theta}) u \right\|_{\Lto}}{\|u\|_{\Lto}},
\end{align*}
which can also be rewritten as~$\left\|\T (\mfq) \right\|^{}_{\Hinf(\Theta)} = \sup\limits_{\theta \in \Theta} \overline{\sigma} \left(\T (\e^{j \theta})\right)$
with~$\overline{\sigma}(\cdot)$ denoting the maximum singular value. 
\end{Def}

As a result, based on Definition~\ref{def:Hinf}, the condition~\eqref{eq: estimation condition} can be equivalently written as
\begin{align}\label{eq:Fault estimation}
    \left\|\T_{f r}(\mfq) - I\right\|^2_{\Hinf(\Theta_m)} \leq \eta_3,~\forall \Theta_m \subset \bar{\Theta}.
\end{align}
As shown in~\eqref{eq:Fault estimation}, the transfer function~$\T_{f r}(\mfq)$ is designed to approximate the identity matrix~$I$ over~$\bar{\Theta}$, so that~$r$ can be viewed as an estimate of~$f$ if~$\T_{f r}(\mfq)$ is sufficiently close to~$I$.
This is different from the system-inversion-based estimation approaches~\cite{dong2011identification,wan2016data} which require~$\T_{f r}(\e^{j\theta})  \equiv I$ (known as the perfect estimation condition). 
We would like to point out that the perfect estimation condition is demanding and generally impossible to achieve because it contains infinite equality constraints, especially when there are disturbances, noise, or unstable zeros.
 
With the condition~\eqref{eq:Fault estimation}, our second problem is to design the fault estimation filter through the following optimization problem, where conditions~\eqref{eq:mapping 1} and~\eqref{eq:mapping 2} are maintained to address~$d$ and~$\omega$, respectively.

\paragraph{{\bf Problem 2} \textup{(Fault estimation filter design)}\bf} 
\textit{
Consider the system~\eqref{eq:SS model}, the filter to be designed in~\eqref{eq: FD filter}, and the expression of the residual~\eqref{eq:residual format}. Given a scalar~$\beta \in [0,1]$, find~$\F(\mfq)$ via the minimization program:
\begin{align*}
    \min_{ \eta_{1},\eta_{3} \in \R_+, ~\F(\mfq)} \quad \{\beta  \eta_{1}+ (1-\beta) \eta_{3}: ~\eqref{eq:mapping 1},~\eqref{eq:mapping 2},~\eqref{eq:Fault estimation} \}.
\end{align*} }

\begin{Rem}[Differences between Problem 1a and 2]
    The condition~\eqref{eq:Fault estimation} for fault estimation is more stringent compared to the condition~\eqref{eq:mapping 3} used for fault detection. 
    In particular, it suffices to let the minimum singular value of $\T_{fr}(\mfq)$ be positive for fault detection, whereas $\T_{fr}(\mfq)$ needs to be as close to~$I$ as possible to obtain satisfactory estimation performance.
    Additionally, filters that satisfy condition~\eqref{eq:Fault estimation} with a sufficiently small~$\Hinf(\Theta)$ norm can provide a positive $\Hmin(\Theta)$ index, but the opposite is not true.
\end{Rem}

%%%%%%%%%%%%%%%%%%%%%%%%%%%%%%%%%%%%%%%%%%%%%%%%%%%%%%%%%%%%%%%%%%%%%
%%%%%%%%%%%%%%%%%%%%%%%%%%%%%%%%%%%%%%%%%%%%%%%%%%%%%%%%%%%%%%%%%%%%%
%                        Main results
%%%%%%%%%%%%%%%%%%%%%%%%%%%%%%%%%%%%%%%%%%%%%%%%%%%%%%%%%%%%%%%%%%%%%
\section{Fault Detection: Optimal Design and Thresholding} \label{sec:FD}
This section presents design methods for the fault detection filter and the thresholding rule that provides guarantees on FAR and FDR. 
To improve the clarity of presentation, some proofs are relegated to Section~\ref{sec:proofs}.

\subsection{Fault detection filter design}
Let us start by considering~$\F(\mfq)$ to be designed in~\eqref{eq: FD filter}. 
In~$\F(\mfq)$, the degrees~$d_N$,~$d_a$, the residual dimension~$n_r$, and coefficients of~$\mathcal{N}(\mfq)$ and~$a(\mfq)$ are all design parameters. For simplicity, let~$n_r$ and~$d_N$ be fixed, and set~$d_N = d_a$ throughout the subsequent analysis.
To compute the $\Hto$ norm and~$\Hmin(\Theta)$ index, the mapping relations~$\T_{\omega r}(\mfq) = -\F(\mfq)W$ and~$\T_{fr}(\mfq) = -\F(\mfq)G$ are represented in the observable canonical forms denoted by~$(\mathcal{A}_r,\mathcal{B}_{\omega r},\mathcal{C}_r)$ and~$(\mathcal{A}_r,\mathcal{B}_{fr},\mathcal{C}_r)$, respectively.
Let $N_{i,j}$ denote the~$j$-th row of~$N_{i}$ for~$i\in\{0,1,\dots,d_N\}$ and~$j\in\{1,\dots,n_r\}$. 
Then, the matrices~$\mathcal{A}_r,~\mathcal{B}_{\omega r},~\mathcal{B}_{fr}$, and~$\mathcal{C}_r$ are given by
\begin{subequations}\label{eq: FD variables}
\begin{align}
    &\mathcal{A}_r = \text{diag}(\underbrace{A_r,\dots,A_r}_{n_r}),
    ~\mathcal{C}_r = \text{diag}(\underbrace{C_r,\dots,C_r}_{n_r}), \\
    &\mathcal{B}_{\omega r} = -[B_{\omega r,1}^{\top},\dots,B_{\omega r,n_r}^{\top}]^{\top},
    ~\mathcal{B}_{f r} = -[B_{fr,1}^{\top},\dots,B_{fr,n_r}^{\top}]^{\top},      
\end{align}
where the subblock matrices are defined as
\begin{align*}
    &A_r = \begin{bmatrix}
            0  &\dots &0 &-a_{0} \\
            1  &\dots &0 &-a_{1} \\
            \vdots  &\ddots &\vdots &\vdots \\
            0  &\dots &1 &-a_{d_N} \\
            \end{bmatrix}, 
    ~B_{\omega r,j} = \begin{bmatrix}
            N_{0,j} \\ N_{1,j} \\ \vdots \\ N_{d_N,j}
        \end{bmatrix} W, 
    ~B_{fr,j} = \begin{bmatrix}
            N_{0,j} \\ N_{1,j} \\ \vdots \\ N_{d_N,j}
      \end{bmatrix} G, ~C_r = \begin{bmatrix} 0 \dots 0~1 \end{bmatrix}.
\end{align*}  
Here, the dimension of the filter state is~$n_{x_r} = n_r(d_N+1)$. The following notations are also introduced for filter design
\begin{align}\label{eq: NbarHbar}
    \bar{\mathcal{N}}= [N_0 ~N_1 ~\dots ~N_{d_N}] ~\text{and} 
    ~\bar{H}= \begin{bmatrix} 
                    H_0 &H_1  &\dots    &0 \\
	                \vdots   &\ddots    &\ddots   &\vdots\\
	                0        &\dots     &H_0 &H_1
               \end{bmatrix}.
\end{align}
\end{subequations} 
Note that the parameters~$a_i$ and~$N_{i}$ to be determined are reformulated into~$\mathcal{A}_r,~\mathcal{B}_{\omega r}$,~$\mathcal{B}_{fr}$, and~$\bar{\mathcal{N}}$. 
An advantage of such a transformation is that all the design parameters are decoupled from each other. 
This allows us to exactly formulate the design of the fault detection filter into a bilinear optimization problem as stated in the following theorem.

\begin{Thm} [Optimal detection: exact finite reformulation]\label{thm:FD filter design}
Consider the system \eqref{eq:SS model}, the structure of the filter \eqref{eq: FD filter}, and the state-space realizations~$(\mathcal{A}_r,\mathcal{B}_{\omega r},\mathcal{C}_r)$ and~$(\mathcal{A}_r,\mathcal{B}_{fr},\mathcal{C}_r)$. 
Given the degree~$d_N$,~$d_a=d_N$, the dimension of the residual~$n_r$, a scalar~$\alpha \in [0,1]$, a sufficiently small~$\vartheta \in \R_+$, and the fault frequency content information~$\bar{\Theta}$, the minimization program in Problem~1a can be equivalently stated as follows 
\begin{subequations}\label{eq:FD opt problem}
\begin{align}
    \min ~&\alpha \eta_1 - (1-\alpha) \eta_2 \notag\\
    \textup{s.t.} ~&\eta_1,~\eta_2 \in \mathbb{R}_+,~a_i\in\R,~N_i\in\R^{n_r\times(n_x+n_y)}, ~i\in\{0,1,\dots,d_N\}, \notag\\
    &P_1 \in \mathcal{S}^{n_{x_r}},~Q_1\in\mathcal{S}^{n_r},~\mathcal{A}_r,~\mathcal{B}_{\omega r},~\mathcal{B}_{f r},~\bar{\mathcal{N}},~\bar{H} ~\text{in} ~\eqref{eq: FD variables}, \notag\\
    ~&P_{2m},~Q_{2m}\in\mathbb{H}^{n_{x_r}}, ~V_m\in\R^{n_{x_r}\times (2n_{x_r}+n_f)}, ~m\in\{1,\dots,n_{\theta}\},   \notag\\
    &\bar{\mathcal{N}}\bar{H} = 0, \label{eq:FD const1}\\
    &\begin{bmatrix}
    P_1 &\mathcal{A}_rP_1  &\mathcal{B}_{\omega r} \\
    *   &P_1 &0\\
    *   &*   &I
    \end{bmatrix} \succeq \vartheta I,
    ~\begin{bmatrix}
    Q_1 &\mathcal{C}_rP_1\\ * &P_1 
    \end{bmatrix} \succeq \vartheta I, ~\textup{Trace}(Q_1)\leq \eta_1-\vartheta, \label{eq:FD const2}\\
    &\begin{bmatrix}
    -P_{2m} &\delta_m Q_{2m}    &0 \\
    *       &\Xi_m    &0\\
    *       &*       &\eta_2 I
    \end{bmatrix} 
    + \begin{bmatrix}
    -I \\ \mathcal{A}^{\top}_r \\\mathcal{B}^{\top}_{fr}
    \end{bmatrix}V_m +V_m^{\top} \begin{bmatrix}     
    -I &\mathcal{A}_r &\mathcal{B}_{fr}\end{bmatrix}    
    \preceq -\vartheta I, \notag \\ 
    & Q_{2m} \succeq \vartheta I, ~m\in\{1,\dots,n_{\theta}\}, \label{eq:FD const3} 
\end{align} 
\end{subequations}
where for each frequency range~$\Theta_m = \{\theta_f : \theta_{1_m} \leq \theta_f \leq \theta_{2_m} \}$, the variables~$\delta_m=\e^{j\theta_{c_m}}$ and~$\Xi_m  = P_{2m}-2\cos(\theta_{d_m})Q_{2m}-\mathcal{C}_r^{\top}\mathcal{C}_r$ with~$\theta_{c_m} = (\theta_{1_m}+\theta_{2_m})/2$
and~$\theta_{d_m} = (\theta_{2_m}-\theta_{1_m})/2$.
\end{Thm}
\begin{proof}
The proof is relegated to Section~\ref{subsec: FD proof}.
\end{proof}

Theorem~\ref{thm:FD filter design} builds on the celebrated GKYP lemma~\cite{iwasaki2005generalized}, which provides three reformulations depending on the desired frequency regimes (low, middle, and high-frequency; see also Lemma~\ref{lem:GKYP} in the proof section). It is worth noting that the assertion of Theorem~\ref{thm:FD filter design} leverages only the middle-frequency part of this lemma, as it covers all the cases required in this study.  
In addition, note that the optimization problem~\eqref{eq:FD opt problem} is nonlinear because of the bilinear terms $\mathcal{A}_rP_1$ in \eqref{eq:FD const2}, and~$\mathcal{A}^{\top}_r V_m$, $\mathcal{B}^{\top}_{fr}V_m$ and their transpose in~\eqref{eq:FD const3}. 
To tackle this issue, the AO method is employed, which divides the decision variables in the bilinear terms into two sets and then optimizes over the two sets of variables alternatively. 
One way of division is 
\begin{align}\label{eq: variable sets}
\begin{split}
    &\mathcal{G}^{k}_1 := \left\{\eta^k_1,\eta^k_2,N^k_i,a^k_i,i\in\{0,1,\dots,d_N\}\right\} ~\text{and} \\ 
    &\mathcal{G}^{k}_2 := \left\{P^k_1, Q^k_1,\eta^k_1, \eta^k_2, P^k_{2m}, Q^k_{2m}, V_m^k, m\in\{1,\dots,n_{\theta}\} \right\},
\end{split}
\end{align} 
where~$k\in N$ serves as the iteration indicator.

The initial values for the optimization process are derived as follows. Initially, a stable denominator, denoted by~$a^0(\mfq)$ with coefficients $a^0_i$, is chosen. 
Next, the coefficients of $\mathcal{N}^0(\mfq)$, i.e.,~$N^0_i$, are determined by solving equation~\eqref{eq:FD const1} subject to the constraint~$\| \bar{\mathcal{N}} \|_{\infty} \geq 1$ to avoid the trivial solution.
Subsequently, the initial values of $\eta_1^0$ and $\eta^0_2$ are found via~\eqref{eq:FD const2} and~\eqref{eq:FD const3}, respectively. 
With these preparations completed, the AO process can be initiated to solve the filter. 
The whole procedure is summarized in Algorithm~\ref{alg:algorithm_1}.

\begin{algorithm}[h] 
	\caption{Solution to the optimization problem~\eqref{eq:FD opt problem}} \label{alg:algorithm_1} 
	\begin{algorithmic}
	\State \textit{Step 1}. \textbf{Initialization of Filter Parameters} 
	\begin{itemize}
	    \item[(a)] Set $d_N,~n_r$, fault frequency ranges~$\bar{\Theta}$, the iteration indicator $k=0$, and select a stable denominator $a^0(\mfq)$ 
	    \item[(b)] Compute $\mathcal{N}^0(\mfq)$ via~\eqref{eq:FD const1} with $\|\bar{\mathcal{N}}\|_{\infty} \geq 1$
	    \item[(c)] Compute $\eta^0_1$ and $\eta^0_2$ via~\eqref{eq:FD const2} and~\eqref{eq:FD const3}, respectively
	\end{itemize}
	\State \textit{Step 2}. \textbf{Optimization of Filter Parameters}
	\begin{itemize}
	    \item[(a)] Select~$\alpha \in [0,1]$, a sufficiently small $\vartheta > 0$
	    \item[(b)] While $|(\alpha\eta^{k+1}_1-(1-\alpha)\eta^{k+1}_2) - (\alpha\eta^{k}_1-(1-\alpha)\eta^{k}_2)| > \vartheta$, do 
	    \begin{itemize}
	        \item[(i)] With $a^k(\mfq)$ and $\mathcal{N}^k(\mfq)$, compute~$P^k_1$ and $V_m^k$ by solving~\eqref{eq:FD opt problem} over $\mathcal{G}^{k}_2$
	        \item[(ii)] With $P^k_1$ and $V_m^k$, compute $a^{k+1}(\mfq)$ and~$\mathcal{N}^{k+1}(\mfq)$ by solving~\eqref{eq:FD opt problem} over $\mathcal{G}^{k}_1$
	        \item[(iii)] Set $k = k+1$ 
	    \end{itemize}
	    \item[(c)] Return final results~$a^{\star}(\mfq)$ and~$\mathcal{N}^{\star}(\mfq)$ 
	\end{itemize}
	\end{algorithmic} 
\end{algorithm}

\begin{Rem}[The auxiliary matrix~$V_m$]
    When using the GKYP lemma to deal with condition~\eqref{eq:mapping 3}, an auxiliary matrices~$V_m$ is introduced to obtain the matrix inequalities in~\eqref{eq:FD const3}. 
    Different from previous results where~$V_m$ is predefined~\cite{wang2008finite,tang2020fault,han2022fault}, it is treated as a decision variable here. 
    This is motivated by the potentially large number of parameters that need determination in~$V_m$ for systems of large scale or dimension. 
    Improper selection of $V_m$ can result in poor $\Hmin(\Theta)$ indices or even render constraints infeasible.
   % By optimizing over~$V_m$ with the AO method, one can achieve better fault sensitivity. 
    Moreover, using relaxation techniques, e.g.,~\cite[Lemma 1]{chang2013new}, to transform~\eqref{eq:FD const3} into linear matrix inequalities easily leads to infeasible problems because multiple constraints restrict the feasible solution set. Therefore, the bilinear terms are retained and addressed using the AO approach.
\end{Rem}

\begin{figure}[t]    
    \centering
    \includegraphics[scale=0.45]{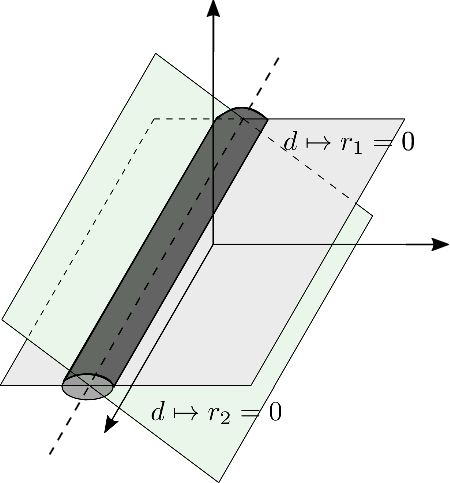} 
    \caption{\small Geometric illustration of the multi-dimensional residual.}
    \label{fig:MultiR}
\end{figure}

\begin{Rem}[Residuals with arbitrary dimensions]
    The proposed design approach enables the fault detection filter to have residuals of arbitrary dimensions.
    Compared to the results~\cite{esfahani2015tractable,pan2019static,pan2021dynamic} also developed in the DAE framework, which generate only one-dimensional residuals, our approach improves two deficiencies: 
\begin{enumerate} [label=(\roman*), itemsep = 0mm, topsep = 0mm, leftmargin = 6mm]
    \item Consider a two-dimensional residual depicted in Fig.~\ref{fig:MultiR} as an example. 
    The filters in~\cite{esfahani2015tractable,pan2019static,pan2021dynamic} cannot detect faults that lie on the same hyperplane as the disturbance, i.e.,~$d \mapsto r = 0$. 
    By considering the two-dimensional residual, faults that can bypass detection only exist at the intersection of two hyperplanes. This means that our approach reduces the size of the set containing undetectable faults;
    
    \item As indicated in~\cite{pan2019static}, different elements of faults may cancel out each other's contributions to the one-dimensional residual. Our approach circumvents this issue by ensuring fault sensitivity with a positive $\Hmin(\Theta)$ index.  
\end{enumerate}
\end{Rem}

\subsection{Thresholding rule}
With the fault detection filter constructed by solving the optimization problem~\eqref{eq:FD opt problem} and the residual evaluation function~$J(r)$ defined in~\eqref{eq: Jr}, the next is to determine the threshold~$J_{th}$ which provides probabilistic guarantees on FAR and FDR as outlined in Problem~1b. 
To proceed, let us first introduce the following lemma and assumption to be used hereafter.

\begin{Lem}[Sub-Gaussian concentration~{\cite[Proposition~2.5.2]{vershynin2018high}}]\label{lem:sub-Gaussian con}
Let $\omega \in \R^{n_{\omega}}$ be subject to a \mbox{sub-Gaussian} distribution with mean $\E[\omega]$ and parameter~$\lambda_{\omega} \in \R_+$, i.e.,
\begin{align*}
    \E\left[\e^{\phi \nu^{\top}(\omega - \E[\omega])} \right] \leq \e^{\lambda^2_{\omega}\phi^2/2},~\forall \phi \in \R~\text{and}~\nu \in \R^{n_{\omega}},
\end{align*}
where~$\|\nu\|_2 = 1$.
Then, the following inequality holds
\begin{align}\label{eq:sub_G con ineq}
    \mPr[\|\omega - \E[\omega]\|_{\infty} \leq \epsilon] \geq 1 - 2n_{\omega}\e^{-\frac{\epsilon^2}{2\lambda^2_{\omega}}}, \quad \forall \epsilon \in \R_+.
\end{align}
\end{Lem}

\begin{As}[Sub-Gaussian noise]\label{As:sub_G}
The measurement noise $\omega$ follows the i.i.d. sub-Gaussian distribution with zero mean and a time-invariant parameter $\lambda_{\omega} \in \R_+$.
\end{As}

The class of sub-Gaussian distributions is board, containing Gaussian, Bernoulli, and all bounded distributions. Also, the tails of sub-Gaussian distributions decrease exponentially fast from~\eqref{eq:sub_G con ineq}, which is expected in many applications.
Given an acceptable FAR, the following theorem provides the determination method of the threshold~$J_{th}$ and FDR.  

\begin{Thm}[Thresholding with probabilistic performance certificates]\label{Thm:performance certificates}
Suppose Assumption~\ref{As:sub_G} holds. Consider the system~\eqref{eq:SS model}, the evaluation function~$J(r)$ in~\eqref{eq: Jr}, the fault detection filter obtained by solving~\eqref{eq:FD opt problem} with the derived values~$\eta^{\star}_1$ and~$\eta^{\star}_2$, and faults of interest~$f \in \Omega_f$. Given an acceptable FAR~$\varepsilon_1 \in (0,1]$, the probabilistic performance~\eqref{eq:FAR prob} in Problem 1b is achieved if the threshold~$J_{th}$ is set as
\begin{align}\label{eq:Jth}
    J_{th} = \lambda_{\omega} \sqrt{2 n_r \eta^{\star}_1 \ln{(2\mathcal{T}n_r / \varepsilon_1)}},
\end{align}
and, when~$\underline{f} > J_{th}\sqrt{n_r/\eta^{\star}_2}$, FDR in~\eqref{eq:FDR prob} satisfies
\begin{small}
\begin{align}\label{eq:FDR}
    ~\mPr\left\{J(r) > J_{th} \big| f \in \Omega_f \right\} \geq 
    \max\left\{0, 1-2\mathcal{T}n_r\e^{-\frac{\left(\underline{f}\sqrt{\eta^{\star}_2/n_r} - J_{th}\right)^2}{2\eta^{\star}_1\lambda^2_{\omega}}}\right\}.
\end{align}
\end{small}
\end{Thm}

\begin{proof}
The proof is relegated to Section \ref{subsec: FD proof}.
\end{proof}

From the concentration property of sub-Gaussian distributions, the threshold~$J_{th}$ in~\eqref{eq:Jth} depends logarithmically on FAR, i.e.,$\sqrt{\ln (1/ \varepsilon_1)}$. This improves the state-of-the-art results (e.g.,~\cite{boem2018plug} and~\cite[Section 10.2.1]{ding2008model}), which rely on Chebyshev's inequality and result in thresholds that scale polynomially with~$\sqrt{1/\varepsilon_1}$. The threshold~\eqref{eq:Jth} also extends our previous work~\cite[Theorem 3.8]{dong2023multimode} where the one-dimensional residual is considered. In addition, a lower bound for~$\underline{f}$ is derived to ensure that FDR can be achieved in~\eqref{eq:FDR}.

\section{Fault Estimation: Optimal Design and Suboptimality Gap}\label{sec:FE}
This section presents design methods for the fault estimation filter and the derivation process of a suboptimality gap for the original estimation problem. 
To improve the clarity of presentation, some proofs are relegated to Section~\ref{sec:proofs}.

\subsection{Fault estimation filter design}
The formulation of the fault estimation filter is provided in~\eqref{eq: FD filter}. 
Based on the desired mapping relations outlined in Problem~2, the design of the filter is formulated into a bilinear optimization problem in the following theorem. 

\begin{Thm}[Optimal estimation: exact finite reformulation]\label{Thm: exact FE}
    Consider the system~\eqref{eq:SS model}, the structure of the filter~\eqref{eq: FD filter}, and the state-space realizations~$(\mathcal{A}_r,\mathcal{B}_{\omega r},\mathcal{C}_r)$ and~$(\mathcal{A}_r,\mathcal{B}_{fr},\mathcal{C}_r)$. Given the filter order~$d_N$, $d_a=d_N$, the dimension of residual~$n_r=n_f$, a scalar~$\beta \in [0,1]$, a sufficiently small~$\vartheta \in \R_+$, and the fault frequency content information~$\bar{\Theta}$, the minimization program in Problem~2 can be equivalently stated as follows 
    \begin{align}\label{eq:FE opt}
    \min ~&\beta\eta_1 + (1-\beta) \eta_3 \notag\\
    \textup{s.t.} ~&\eta_1,\eta_3 \in \mathbb{R}_+,a_i\in\R,~N_i\in\R^{n_r\times(n_x+n_y)}, ~i\in\{0,1,\dots,d_N\}, \notag\\
    &P_1 \in \mathcal{S}^{n_{x_r}},Q_1\in\mathcal{S}^{n_r},\mathcal{A}_r,~\mathcal{B}_{\omega r},~\mathcal{B}_{f r},~\bar{\mathcal{N}},~\bar{H} ~\text{in} ~\eqref{eq: FD variables}, \notag\\
    ~&P_{2m},~Q_{2m}\in\mathbb{H}^{n_{x_r}},V_m\in\R^{n_{x_r}\times (2n_{x_r}+n_f)}, m\in\{1,\dots,n_{\theta}\}, \notag\\
    &\eqref{eq:FD const1}, ~\eqref{eq:FD const2}, \notag\\
    &\begin{bmatrix}
    -P_{2m}  &\delta_m Q_{2m}    &0 \\
    *       &\bar{\Xi}_m    &-\mathcal{C}_r^{\top}\\
    *      &*       &I-\eta_3 I
    \end{bmatrix} 
    + \begin{bmatrix}
    -I \\ \mathcal{A}^{\top}_r \\\mathcal{B}^{\top}_{fr}
    \end{bmatrix}V_m +V_m^{\top}\begin{bmatrix}-I &\mathcal{A}_r &\mathcal{B}_{fr}\end{bmatrix}
    \preceq -\vartheta I, \notag\\
    &Q_{2m} \succeq \vartheta I, ~m\in\{1,\dots,n_{\theta}\},
    \end{align} 
    where for each frequency range~$\Theta_m = \{\theta_f: \theta_{1_m} \leq \theta_f \leq \theta_{2_m} \}$, the variables $\delta_m=\e^{j\theta_{c_m}}$,~$\bar{\Xi}_m  = P_{2m}-2\cos(\theta_{d_m})Q_{2m}+\mathcal{C}_r^{\top}\mathcal{C}_r$ with $\theta_{c_m} = (\theta_{1_m}+\theta_{2_m})/2$ and~$\theta_{d_m} = (\theta_{2_m}-\theta_{1_m})/2$.
\end{Thm}
\begin{proof}
    It is proved in Theorem~\ref{thm:FD filter design} that~\eqref{eq:FD const1} and~\eqref{eq:FD const2} are equivalent to conditions~\eqref{eq:mapping 1} and~\eqref{eq:mapping 2}, respectively.   
    To demonstrate the equivalence between constraints~\eqref{eq:FE opt} and conditions~\eqref{eq:Fault estimation}, the state-space realization of~\mbox{$\T_{f r}(\mfq)-I$} is derived as~$(\mathcal{A}_r,\mathcal{B}_{fr},\mathcal{C}_r,-I)$. 
    By setting the matrix~$\Pi = \textup{diag}(I,-\eta_3 I)$ and using~$(\mathcal{A}_r,\mathcal{B}_{fr},\mathcal{C}_r,-I)$ in Lemma~\ref{lem:GKYP}, the equivalence between~\eqref{eq:FE opt} and~\eqref{eq:Fault estimation} is established. 
    The proof procedure of the equivalence is similar to that of~\eqref{eq:FD const3} in the proof of Theorem~\ref{thm:FD filter design}. This completes the proof.
\end{proof}

The optimization problem in Theorem~\ref{Thm: exact FE} can be solved using Algorithm~\ref{alg:algorithm_1} as well. 
However, the key to achieving satisfactory estimation results is to ensure that~$\|\T_{f r}(\mfq)-I\|_{\Hinf(\Theta_m)}$ is sufficiently small. 
This usually requires several iteration steps with Algorithm~\ref{alg:algorithm_1} and results in heavy computational loads when dealing with large-scale systems. 

\subsection{Convex approximation with suboptimality gap}
To reduce the computational complexity, the estimation condition~\eqref{eq:Fault estimation} is relaxed by letting~$\T_{fr}(\mfq)$ approximate the identity matrix at~$\kappa \in \N$ selected finite frequency points~$\theta_i \in \bar{\Theta}$ instead of considering all frequencies, i.e.,
\begin{align}\label{eq: relaxed FE con}
    \left\| \T_{fr}(\e^{j\theta_i})-I \right\|^2_{2} \leq \bar{\eta}_{3}, \quad  \forall i \in\{1,\dots,\kappa\},
\end{align}
where~$\bar{\eta}_{3} \in \R_+$. The relaxed version of Problem~2 is derived as follows.

\paragraph{{\bf Problem 2r} \textup{(Fault estimation with finite frequency content)}\bf} 
\textit{
Consider the system~\eqref{eq:SS model}, the filter to be designed in~\eqref{eq: FD filter}, and the expression of the residual~\eqref{eq:residual format}. Given a scalar~$\beta \in [0,1]$, find~$\F(\mfq)$ via the minimization program:
\begin{align*}
        \min_{ \eta_1,\bar{\eta}_{3}  \in \R_+, \F(\mfq)} ~\left\{\beta  \eta_{1}+ (1-\beta) \bar{\eta}_{3} : \eqref{eq:mapping 1},\eqref{eq:mapping 2},\eqref{eq: relaxed FE con} \right\}.
\end{align*}}

Before presenting the solution to Problem 2r, let us make some clarifications on~$\F(\mfq)$.
For simplicity, the poles of the filter are fixed. Specifically, roots of~$a(\mfq)$ are selected inside the unit disk and the order is set as~$d_a=d_N$, so that the fault estimation filter is stable and strictly proper. 
Thus, the coefficient matrices~$N_i$ for~$i \in \{0,1,\dots,d_N\}$ become the only parameters to be determined.
For clarity, by using the multiplication rule of polynomial matrices~\cite[Lemma 4.2]{esfahani2015tractable}, the transfer functions~$\T_{fr}(\mfq)$ and~$\T_{\omega r}(\mfq)$ outlined in~\eqref{eq:residual format} are written as
\begin{align}
    \T_{fr}(\mfq) = -\frac{\mathcal{N}(\mfq)G}{a(\mfq)} 
                  =  \bar{\mathcal{N}}\Psi_G(\mfq)~\text{and} 
    ~\T_{\omega r}(\mfq) = -\frac{\mathcal{N}(\mfq)W}{a(\mfq)} 
                        =  \bar{\mathcal{N}}\Psi_W(\mfq), 
\label{eq:Tran f2fest}
\end{align}
where
\begin{align*}
    &\Psi_G(\mfq) = -a^{-1}(\mfq) \text{diag}(G,\dots,G) [I,\mfq I, \dots, \mfq^{d_N}I]^{\top} ~\text{and} \\
    &\Psi_W(\mfq) =-a^{-1}(\mfq) \text{diag}(W,\dots,W) [I,\mfq I, \dots, \mfq^{d_N}I]^{\top}.
\end{align*}
Subsequently, the design method of the fault estimation filter with relaxed conditions depicted in Problem 2r is provided in the following theorem. 

\begin{Thm}[Optimal estimation: finite relaxation QP]\label{Thm:FE filter design}
    Consider the system~\eqref{eq:SS model}, the structure of the filter~\eqref{eq: FD filter}, and the reformulations of~$\T_{fr}(\mfq)$ and~$\T_{\omega r}(\mfq)$ in~\eqref{eq:Tran f2fest}. Given the order~$d_N$, the dimension~$n_r=n_f$, the stable denominator~$a(\mfq)$ with~$d_a = d_N$, ~$\kappa$ frequency points~$\theta_i \in \bar{\Theta}$, and the weight~$\beta \in [0,1]$, the optimization problem in Problem 2r can be reformulated as the following QP problem:
    \begin{subequations}\label{eq:fest opt}
    \begin{align}
    \min ~&\beta \eta_1+ (1-\beta) \bar{\eta}_{3} \notag \\ 
    \textup{s.t.} 
    ~&\bar{\mathcal{N}} ,\bar{H} ~\text{in}~\eqref{eq: NbarHbar}, ~\eta_{1},~\bar{\eta}_{3} \in \R_+, \notag \\
    & \bar{\mathcal{N}} \bar{H} = 0, \label{eq:fest opt1}\\
    & \textup{Trace}\left[\bar{\mathcal{N}} \Phi \bar{\mathcal{N}}^{\top}\right] \leq \eta_1, \label{eq:fest opt2} \\
    &\left\|\begin{bmatrix}
        \bar{\mathcal{N}}\mathcal{R}_i - I &-\bar{\mathcal{N}}\mathcal{I}_i\\
        \bar{\mathcal{N}}\mathcal{I}_i &\bar{\mathcal{N}}\mathcal{R}_i - I
    \end{bmatrix}\right\|^2_2 \leq \bar{\eta}_{3},~\forall i \in \{1,\dots,\kappa\} \label{eq:fest opt3}
    \end{align}
    \end{subequations} 
where~$\mathcal{R}_i = \textup{Real}\left(\Psi_G(\e^{j\theta_i})\right)$ and $\mathcal{I}_i = \textup{Imag}\left(\Psi_G(\e^{j\theta_i})\right)$ are the real and imaginary parts of $\Psi_G(\e^{j\theta_i})$, respectively, and~$\Phi = \frac{1}{2\pi} \int^{\pi}_{-\pi} \Psi_W(\e^{j\theta})  \Psi_W^{*}(\e^{j\theta}) \diff \theta$.
\end{Thm}
\begin{proof}
The proof is relegated to Section \ref{subsec:FE proof}.
\end{proof}

Compared to~\eqref{eq:FE opt}, the design of the fault estimation filter presented in Problem 2r stands out for its integration of more lenient conditions, as expounded in reference~\eqref{eq:fest opt}. Notably, this design exhibits computational tractability, owing to its formulation as a QP problem.
In addition, an approximate analytical solution to~\eqref{eq:fest opt} is given as follows.

\begin{Cor}[Approximate analytical solution]\label{cor:Analytical solution}
Consider the QP problem in~\eqref{eq:fest opt} with the $2$ norm replaced by the Frobenius norm. An approximate analytical solution to~\eqref{eq:fest opt} is: 
    \begin{align}
    \bar{\mathcal{N}}_{App}^{\star} 
    =  \begin{bmatrix}
        \frac{4(1-\beta)}{\kappa}\sum^{\kappa}_{i=1} \mathcal{R}^{\top}_i &0
    \end{bmatrix} 
    \begin{bmatrix}
        2\beta  \Phi +\frac{4(1-\beta)}{\kappa}\sum^{\kappa}_{i=1}\left( \mathcal{R}_i\mathcal{R}^{\top}_i  +  \mathcal{I}_i\mathcal{I}^{\top}_i\right) &\bar{H} \\
        \bar{H}^{\top} &0
    \end{bmatrix}^{\dagger} \begin{bmatrix}
        I \\0
    \end{bmatrix} , \label{eq: ApproSol}
\end{align}
where~$(\cdot)^{\dagger}$ denotes the pseudo-inverse. 
\end{Cor}

\begin{proof}
The proof is relegated to Section \ref{subsec:FE proof}.
\end{proof}

It is worth mentioning that, for a filter with given poles (fixed denominator~$a(\mfq)$), a suboptimality gap for the original estimation problem stated in Problem 2 can be obtained by solving the optimization problems in Theorem~\ref{Thm: exact FE} and Theorem~\ref{Thm:FE filter design}. 
This result is presented in Proposition~\ref{prop: opt gap}.
To enhance readability, let us denote the optimal value of the objective function in Problem 2 as~$\mathcal{J}^{\star}$ with a given denominator~$a(\mfq)$, i.e.,
\begin{align*}
    \mathcal{J}^{\star}  = \min_{\mathcal{N}(\mfq)} \left\{ \beta \|\T_{\omega r}(\mfq) \|^2_{\mathcal{H}_2} + (1-\beta) \|\T_{fr}(\mfq) - I\|^2_{\mathcal{H}_{\infty}(\Theta)} :\T_{Xr}(\mfq) = 0 \right\}.
\end{align*}
Furthermore, Let~$\eta^{\star}_{1,AO}$ and $\eta^{\star}_{3,AO}$ denote the results obtained by solving the optimization problem~\eqref{eq:FE opt} using the AO approach.
Use~$\eta^{\star}_{1,RR}$ and~$\bar{\eta}^{\star}_{3,RR}$ to denote the optimal values obtained by solving the optimization problem~\eqref{eq:fest opt}.
Subsequently, the suboptimality gap for Problem 2 is presented in the next proposition.

\begin{Prop}[Suboptimality gap with fixed poles] \label{prop: opt gap}
     Given a stable denominator~$a(\mfq)$, the optimal value of the objective function in Problem 2 is bounded by
    \begin{align}\label{eq: opt gap}
       \beta \eta^{\star}_{1,RR} + (1-\beta) \bar{\eta}^{\star}_{3,RR}  
       \leq   \mathcal{J}^{\star} 
       \leq \beta \eta^{\star}_{1,AO} + (1-\beta) \eta^{\star}_{3,AO}.
    \end{align}
\end{Prop}

\begin{proof}
The proof is relegated to Section \ref{subsec:FE proof}.
\end{proof}

In contrast to the immediate acquisition of the lower bound from the optimization problem's resolution in reference~\eqref{eq:fest opt}, the upper bound derived through the AO approach generally demands multiple iterative phases. This iterative nature can lead to substantial computational burdens unless the initial value is judiciously selected. Fortunately, a remedy lies in employing the solution from the more lenient design problem described in Theorem~\ref{Thm:FE filter design} as the starting point. This initial solution provides a solid foundation for refining the upper bound outlined in reference~\eqref{eq:FE opt} through the utilization of the AO approach in solving the optimization problem. The entire process is succinctly encapsulated in Algorithm~\ref{alg:algorithm_2}.

\begin{algorithm}[h]
	\caption{Computing the suboptimality gap in~\eqref{eq: opt gap}} \label{alg:algorithm_2} 
	\begin{algorithmic}
        \State \textit{Setp 1.} \textbf{Initialization} 
        \begin{itemize}
           \item[(a)] Select $d_N$,~$n_r = n_f$, and a stable denominator $a(\mfq)$ 
           \item[(b)] Select $\kappa$ frequency points uniformly from the frequency range~$\Theta$ and the weight $\beta$
        \end{itemize}        
	\State \textit{Step 2.} \textbf{Derivation of the lower bound} 
	\begin{itemize}
        \item[(a)] Compute the matrix~$\mathcal{R}_i$, $\mathcal{I}_i$, and~$\Phi$ for $i\in \{1,\dots,\kappa\}$
	    \item[(b)] Find the numerator $\mathcal{N}^{\star}_{RR}(\mfq)$ and the bounds~$\eta^{\star}_{1,RR}$ and~$\bar{\eta}^{\star}_{3,RR}$ by solving~\eqref{eq:fest opt}
	    \item[(c)] Output the lower bound: $\beta \eta^{\star}_{1,RR} + (1-\beta) \bar{\eta}^{\star}_{3,RR}$ 
	\end{itemize}
	\State \textit{Step 3}. \textbf{Derivation of the upper bound}
	\begin{itemize}
	    \item[(a)] Set $\mathcal{N}^{\star}_{RR}(\mfq)$ as the initial condition and fix $a(\mfq)$ for~\eqref{eq:FE opt}
	    \item[(b)] Optimize the numerator by solving~\eqref{eq:FE opt} with the AO approach, and obtain~$\eta^{\star}_{1,AO}$ and $\eta^{\star}_{3,AO}$ 
	    \item[(c)] Output the upper bound: $\beta \eta^{\star}_{1,AO} + (1-\beta) \eta^{\star}_{3,AO}$
	\end{itemize}
	\end{algorithmic} 
\end{algorithm}

This section is closed with the following remarks on the proposed design approaches to fault estimation filters.
\begin{Rem}[Trade-off analysis]
    There is a trade-off between decoupling the unknown signals $X$ (consisting of the unknown state~$x$ and disturbance~$d$), suppressing the noise~$\omega$, and estimating the fault~$f$ in~\eqref{eq:FE opt} and~\eqref{eq:fest opt}. First, the feasible solutions to~\eqref{eq:FE opt} and~\eqref{eq:fest opt} lie in the left null space of~$\bar{H}$, which restricts the choice of~$\bar{\mathcal{N}}$. 
    Second, increasing~$\beta$ improves the noise suppression capability of the filter. However, it reduces the estimation performance and vice versa. 
    The trade-offs can, therefore, be used as a guide for selecting appropriate weights.
\end{Rem}

\begin{Rem}[Selection of decision variable sets] 
    When using the AO approach to solve the bilinear optimization problems stated in Theorem~\ref{thm:FD filter design} and Theorem~\ref{Thm: exact FE}, it is essential to partition the decision variables in the bilinear terms into two sets, namely $\mathcal{G}^k_1$ and $\mathcal{G}^k_2$.
    We observe that, for different optimization problems, the choice of decision variable sets greatly influences the convergence speed of the AO approach. 
    In particular, when solving the optimization problem~\eqref{eq:FE opt}, if the decision variable sets are selected without overlap, i.e., $\{\eta^k_1, \eta_2^k, N^k_i,a^k_i,i\in\{0,1,\dots,d_N\}\}$ and $\{P^k_1, Q^k_1, P^k_{2m}, Q^k_{2m}, V_m^k,m\in\{1,\dots,n_{\theta}\}\}$, it leads to a more efficient solution compared to the way in~\eqref{eq: variable sets}.
\end{Rem}

\begin{Rem}[Fault estimation for non-minimum phase systems]
    For non-minimum phase systems, it is reported that the optimal distance between~$\T_{fr}$ and $I$ in the~$\mathcal{H}_{\infty}$ framework is~$1$~\cite[Theorem 14.5]{ding2008model}, i.e., ~$\min_{\bar{\mathcal{N}}} \|\T_{fr}(\mfq) - I \|_{\mathcal{H}_\infty} = 1$, which indicates that a satisfactory fault estimation over the whole frequency range is not achievable.
    Our methods proposed in Theorem~\ref{Thm: exact FE} and Theorem~\ref{Thm:FE filter design} can improve the estimation performance by limiting the frequency ranges.
    This assertion will be substantiated by supporting evidence from simulation results.
\end{Rem}

\begin{Rem}[Non-decoupled disturbances with frequency content information]
    For disturbances that cannot be completely decoupled, and supposing that the knowledge of disturbance frequency content is available, the restricted $\mathcal{H}_{\infty}(\Theta)$ norm can be employed to limit their impact on residuals.
    It is observed from off-line exhaustive simulations that expanding the frequency range of disturbances does not significantly affect fault sensitivity, while the ability to suppress disturbances degrades.
\end{Rem}

\begin{Rem}[Conservatism analysis] The conservatism of the fault estimation filter design method is summarized as follows:
\begin{enumerate} [label=(\roman*), itemsep = 0mm, topsep = 0mm, leftmargin = 6mm]
    \item To reduce computational complexity, a selective approach is adopted for the design of fault estimation filters in~\eqref{eq:fest opt}, where constraints are only imposed on a subset of frequency points in~$\bar{\Theta}$. 
    As a result, the estimation performance at the other frequency points in~$\bar{\Theta}$ may not be guaranteed. 
    However, as demonstrated by simulation results, the degradation of estimation performance at those points is minor.
    \item For simplicity, the denominator of the transfer function~$a(\mfq)$ is fixed in the optimization problem~\eqref{eq:fest opt}, which restricts the design freedom. However, including the simultaneous design of both~$a(\mfq)$ and $\mathcal{N}(\mfq)$ would result in a much more complex optimization problem, which might not be computationally tractable.
\end{enumerate}
\end{Rem}

%%%%%%%%%%%%%%%%%%%%%%%%%%%%%%%%%%%%%%%%%%%%%%%%%%%%%%%%%%%%%%%%%%%%%%%%%%%
%%%%%%%%%%%%%%%%%%%%%%%%%%%%%%%%%%%%%%%%%%%%%%%%%%%%%%%%%%%%%%%%%%%%%%%%%%%
\section{Technical Proofs of Main Results} \label{sec:proofs}

\subsection{Proofs of results in fault detection}\label{subsec: FD proof}
The following two lemmas are required for the proof of Theorem~\ref{thm:FD filter design}. 

\begin{Lem}[GKYP lemma \cite{iwasaki2005generalized}] \label{lem:GKYP}
Consider a transfer function defined as $\T(\mfq) =  \mathcal{C}(\mfq I- \mathcal{A})^{-1}\mathcal{B}+\mathcal{D}$. Given a symmetric matrix $\Pi$ and a frequency range~$\Theta$, the following statements are equivalent:
\begin{enumerate}[label=(\roman*), itemsep = 0mm, topsep = 0mm, leftmargin = 6mm]
    \item The inequality holds in the frequency range~$\theta \in \Theta$
    \begin{align*}
        \begin{bmatrix}
        \T(\e^{j\theta}) \\I
        \end{bmatrix}^{*} \Pi
        \begin{bmatrix}
        \T(\e^{j\theta}) \\I
        \end{bmatrix} \prec 0.
    \end{align*}
    \item There exist Hermitian matrices $\mathcal{P}$ and~$\mathcal{Q}$ with appropriate dimensions and $\mathcal{Q} \succ 0$ such that
    \begin{equation*}
        \begin{bmatrix} \mathcal{A} &\mathcal{B} \\ I &0 \end{bmatrix}^{\top} \Lambda
        \begin{bmatrix} \mathcal{A} &\mathcal{B} \\ I &0 \end{bmatrix} + 
        \begin{bmatrix} \mathcal{C} &\mathcal{D} \\ 0 &I
        \end{bmatrix}^{\top} \Pi
        \begin{bmatrix} \mathcal{C} &\mathcal{D} \\ 0 &I \end{bmatrix} \prec 0,
    \end{equation*} 
    where the following hold: \\
        a. For the low frequency range~$\Theta = \{\theta : 0 \leq \theta \leq \theta_l \}$,~$\Lambda = \begin{bmatrix}
        -\mathcal{P} &\mathcal{Q} \\
        \mathcal{Q} &\mathcal{P}-2\cos(\theta_l)\mathcal{Q}
        \end{bmatrix}$;  \\
        b. For the middle frequency range $\Theta = \{\theta : \theta_1 \leq \theta \leq \theta_2 \}$,~$\Lambda = \begin{bmatrix}
        -\mathcal{P} & \e^{j\theta_c} \mathcal{Q} \\
        \e^{-j\theta_c}\mathcal{Q} &\mathcal{P}-2\cos(\theta_d)\mathcal{Q}
        \end{bmatrix}$ , where $\theta_c = (\theta_1+\theta_2)/2$ and~$\theta_d = (\theta_2-\theta_1)/2$; \\
        c. For the high frequency range~$\Theta = \{\theta :\theta_h \leq \theta \leq \pi   \}$,~$\Lambda = \begin{bmatrix}
        -\mathcal{P} &-\mathcal{Q} \\
        -\mathcal{Q} &\mathcal{P}+2\cos(\theta_h)\mathcal{Q}
        \end{bmatrix}$ .
\end{enumerate}
\end{Lem}

\begin{Lem}[Finsler's lemma \cite{boyd1994linear}] \label{lem:Finsler}
For matrices~$\mathcal{V}$ and~$\mathcal{Y}$ with appropriate dimensions, the following statements are equivalent:
\begin{enumerate}[label=(\roman*), itemsep = 0mm, topsep = 0mm, leftmargin = 6mm]
    \item $\mathcal{Y}^{\perp}\mathcal{V} \left(\mathcal{Y}^{\perp}\right)^{\top} \prec 0$, where $\mathcal{Y}^{\perp}$ denote the matrix satisfying $\mathcal{Y}^{\perp}\mathcal{Y} = 0$;
    \item There exists a matrix $\mathcal{U}$ such that $\mathcal{V} + \mathcal{Y}\mathcal{U} + \mathcal{U}^{\top}\mathcal{Y}^{\top} \prec 0$.
\end{enumerate}
\end{Lem}

\begin{proof}[{\textbf{Proof of Theorem~\ref{thm:FD filter design}}}]
First, according to the multiplication rule of polynomial matrices~\cite[Lemma 4.2]{esfahani2015tractable}, the constraint~\eqref{eq:FD const1} implies~$\mathcal{N}(\mfq)H(\mfq) = 0$, which means that~$X$ is completely decoupled from~$r$. Thus,~\eqref{eq:mapping 1} is satisfied. 
Second, from the expression of $r$ in~\eqref{eq:residual format}, the transfer function from~$\omega$ to~$r$ is~$-a^{-1}(\mfq)\mathcal{N}(\mfq)W$ when~\eqref{eq:FD const1} is satisfied, and its state-space realization is denoted as~$(\mathcal{A}_r,\mathcal{B}_{\omega r},\mathcal{C}_r)$. 
According to the classical result on~$\Hto$ norm~\cite[Lemma 1]{de2002extended}, the equivalence between~\eqref{eq:FD const2} and~\eqref{eq:mapping 2} can be obtained directly.

In the last part of the proof, the equivalence between~\eqref{eq:FD const3} and the mapping relation~\eqref{eq:mapping 3} for a single frequency range~$\Theta_m$ is established.
According to Lemma~\ref{lem:Finsler}, the first matrix inequality in~\eqref{eq:FD const3} is equivalent to 
\begin{align*}
    \begin{bmatrix}
    \begin{bmatrix}
    \mathcal{A}_r^{\top} \\ \mathcal{B}_{fr}^{\top}
    \end{bmatrix} &I
    \end{bmatrix}
    \left[\begin{array}{c|cc}
    -P_{2m}  &\delta_m Q_{2m}  &0 \\ \hline
    *        &\Xi_m  &0\\
    *        &*    &\eta_2 I
    \end{array} \right]
    \begin{bmatrix}
    \begin{bmatrix}
    \mathcal{A}_r &\mathcal{B}_{fr}
    \end{bmatrix} \\I
    \end{bmatrix} \preceq -\vartheta I,
\end{align*}
where~$\delta_m=\e^{j\theta_{c_m}}$ and~$\Xi_m  = P_{2m}-2\cos(\theta_{d_m})Q_{2m}-\mathcal{C}_r^{\top}\mathcal{C}_r$ . The above inequality can be expanded into
\begin{align}
    &\begin{bmatrix}
    \Xi_m  &0\\ 0  &\eta_2 I
    \end{bmatrix} - 
    \begin{bmatrix}
    \mathcal{A}_r^{\top} \\ \mathcal{B}_{fr}^{\top}
    \end{bmatrix}P_{2m}
    \begin{bmatrix}
    \mathcal{A}_r &\mathcal{B}_{fr}
    \end{bmatrix} +
    \begin{bmatrix}
    \mathcal{A}_r^{\top} \\ \mathcal{B}_{fr}^{\top}
    \end{bmatrix}
    \begin{bmatrix}  \e^{j\theta_{c_m}} Q_{2m} &0  \end{bmatrix} 
    +
    \begin{bmatrix}  \e^{-j\theta_{c_m}}Q_{2m} \\0    \end{bmatrix}
    \begin{bmatrix}
    \mathcal{A}_r &\mathcal{B}_{fr}
    \end{bmatrix} \notag\\
    =&\begin{bmatrix}
        \Xi_m-\mathcal{A}_r^{\top}P_{2m}\mathcal{A}_r+\e^{j\theta_{c_m}}\mathcal{A}_r^{\top}Q_{2m}+\e^{-j\theta_{c_m}}Q_{2m}\mathcal{A}_r &-\mathcal{A}_r^{\top}P_{2m}\mathcal{B}_{fr} + \e^{-j\theta_{c_m}}Q_{2m}\mathcal{B}_{fr} \\
        * &-\mathcal{B}_{fr}^{\top}P_{2m}\mathcal{B}_{fr}+\eta_2 I
    \end{bmatrix} \notag\\
    =&\begin{bmatrix} \mathcal{A}_r &\mathcal{B}_{fr}\\ I &0   \end{bmatrix}^{\top}
    \begin{bmatrix} -P_{2m} &\e^{j\theta_{c_m}}Q_{2m} \\ * &P_{2m}-2\cos(\theta_{{d_m}})Q_{2m}  \end{bmatrix}
    \begin{bmatrix} \mathcal{A}_r &\mathcal{B}_{fr}\\ I &0   \end{bmatrix} 
    + 
    \begin{bmatrix} \mathcal{C}_r &0 \\0 &I \end{bmatrix}^{\top}
    \begin{bmatrix}    -I &0 \\ 0 &\eta_2 I    \end{bmatrix}
    \begin{bmatrix} \mathcal{C}_r &0 \\0 &I \end{bmatrix} \notag \\
    \preceq &-\vartheta I. \label{eq:Proof-ineq2}
\end{align}
Recall that the transfer function from $f$ to $r$, denoted by ~$\T_{fr}(\mfq)$, has a state-space realization given by~$(\mathcal{A}_r,\mathcal{B}_{f r},\mathcal{C}_r)$.  
According to the middle-frequency case in Lemma~\ref{lem:GKYP}, the last equation of~\eqref{eq:Proof-ineq2} is equivalent to
\begin{align*}
        \begin{bmatrix} \T_{fr}(\e^{j\theta}) \\I \end{bmatrix}^{*} 
        \begin{bmatrix} -I &0 \\0 &\eta_2 I\end{bmatrix}
        \begin{bmatrix}\T_{fr}(\e^{j\theta}) \\I \end{bmatrix}
        \preceq -\vartheta I.
\end{align*}
Thus, it holds that $\|\T_{fr}(\e^{j\theta})\|^2_{\Hmin(\Theta_m)} \geq \eta_2$ for $\theta \in \Theta_m$. 
This completes the proof.
\end{proof}

The following lemma is introduced to prove Theorem~\ref{Thm:performance certificates}.

\begin{Lem}[Linear transformation of sub-Gaussian signals~{\cite[Lemma 4.3]{dong2023multimode}}] \label{lem:Linear trans}
Let $\T_{\omega r}$ be the transfer function from~$\omega$ to $r$. If $\omega$ follows the i.i.d. sub-Gaussian distribution with zero mean and parameter $\lambda_{\omega}$, the signal $r$ is also sub-Gaussian with zero mean and the respective parameter~$\lambda_r = \|\T_{\omega r}\|_{\Hto} \lambda_{\omega}$.
\end{Lem}

\begin{proof}[\textbf{Proof of Theorem~\ref{Thm:performance certificates}}]
Let us first show that the given FAR~$\varepsilon_1$ is guaranteed if $J_{th}$ is determined by~\eqref{eq:Jth} in the absence of faults. 
From the expression of the residual~\eqref{eq:residual format}, $r = \T_{\omega r}(\mfq) [\omega]$ since $X$ is decoupled and~$f=0$. 
According to Lemma~\ref{lem:Linear trans}, $r$ is sub-Gaussian and its parameter~$\lambda_r$ satisfies
\begin{align}\label{eq: lambda r}
    \lambda_r = \|\T_{\omega r}(\mfq)\|_{\Hto} \lambda_{\omega} \leq \sqrt{\eta^{\star}_1} \lambda_{\omega},
\end{align}
where~\eqref{eq: lambda r} holds by invoking Theorem~\ref{thm:FD filter design}. Then, we have
\begin{align*}
    \mPr[J(r) > J_{th}|f=0] 
    &= \mPr\left[ \frac{1}{\mathcal{T}}  \sum^{k_1+\mathcal{T}}_{k=k_1} \|r(k)\|_2 > J_{th} \bigg|f=0  \right] \notag\\
    &\stackrel{(a)}{\leq} \mPr\left[ \sum^{k_1+\mathcal{T}}_{k=k_1} \sqrt{n_r} \|r(k)\|_{\infty} > \mathcal{T}J_{th} \bigg|f=0  \right] \notag\\
    &\stackrel{(b)}{\leq} \sum^{k_1+\mathcal{T}}_{k=k_1}\mPr\left[\|r(k)\|_{\infty} > \frac{J_{th}}{\sqrt{n_r}} \bigg|f=0 \right] \notag\\
    &\stackrel{(c)}{\leq} 2\mathcal{T}n_r \e^{-\frac{(J_{th}/\sqrt{n_r})^2}{2  \lambda^2_{r} }}
    \stackrel{(d)}{\leq} 2\mathcal{T}n_r \e^{-\frac{J^2_{th}}{2 n_r \eta^{\star}_1 \lambda^2_{\omega} }}.
\end{align*}
The inequality (a) holds as a result of the equivalence between vector norms, i.e.,~$\|r(k)\|_2 \leq \sqrt{n_r}\|r(k)\|_{\infty}$. The inequality (b) holds due to the fact that~$\mPr[v_1 + v_2 > v_3] \leq \mPr[v_1  > v_3/2] + \mPr[v_2  > v_3/2]$ where~$v_1,v_2,v_3\in\R_+$.
The inequality~(c) is derived from the concentration inequality in Lemma~\ref{lem:sub-Gaussian con}. And the inequality (d) is obtained according to \eqref{eq: lambda r}. 
Substituting~\eqref{eq:Jth} into the last inequality yields~$\mPr[J(r) > J_{th}|f=0] \leq \varepsilon_1$.
This completes the first part of the proof. 

The second step is to demonstrate that~\eqref{eq:FDR} holds for~$f\in\Omega_f$. Consider the residual~$r = \T_{fr}[f] + \T_{\omega r} [\omega]$ in the presence of faults, whose expectation is~$\E[r] = \T_{fr}[f]$. 
Note that~$r - \E[r] = \T_{\omega r} [\omega]$ is sub-Gaussian with the parameter~$\sqrt{\eta^{\star}_1} \lambda_{\omega}$ as indicated above. Thus, for a positive scalar~$\epsilon \in \R_+$, it holds that 
\begin{align*}
    \mPr\left\{ \sum^{k_1+\mathcal{T}}_{k=k_1} \|r(k)-\E[r(k)]\|_{\infty} 
       > \mathcal{T} \epsilon \bigg| f \in \Omega_f \right\} 
    \leq 2\mathcal{T}n_r\e^{-\frac{\epsilon^2}{2\eta^{\star}_1\lambda^2_{\omega}}},
\end{align*}
which is equivalent to 
\begin{align*}
    \mPr\left\{  \sum^{k_1+\mathcal{T}}_{k=k_1} \|r(k)-\E[r(k)]\|_{\infty} 
       \leq \mathcal{T} \epsilon \bigg| f \in \Omega_f \right\} 
    \geq 1-2\mathcal{T}n_r\e^{-\frac{\epsilon^2}{2\eta^{\star}_1\lambda^2_{\omega}}}.
\end{align*}
Since~$\sum^{k_1+\mathcal{T}}_{k=k_1} (\|\E[r(k)]\|_{\infty} 
     -  \| r(k)\|_{\infty} )
     \leq \sum^{k_1+\mathcal{T}}_{k=k_1} \|r(k) - \E[r(k)]\|_{\infty}$,
we have
\begin{align*}
\mPr\left\{\sum^{k_1+\mathcal{T}}_{k=k_1} (\|\E[r(k)]\|_{\infty} 
     -  \| r(k)\|_{\infty} )  \leq \mathcal{T}\epsilon \bigg| f \in \Omega_f \right\} 
    \geq 1-2\mathcal{T}n_r\e^{-\frac{\epsilon^2}{2\eta^{\star}_1\lambda^2_{\omega}}}.
\end{align*}

\noindent Let $\mathcal{T}\epsilon = \sum^{k_1+\mathcal{T}}_{k=k_1}\|\E[r(k)]\|_{\infty} - \mathcal{T}J_{th} > 0$. The above inequality becomes
\begin{align}\label{eq:FDR ineq1}
   \mPr\left\{ \sum^{k_1+\mathcal{T}}_{k=k_1} \left\| r(k)\right\|_{\infty} \geq \mathcal{T}J_{th} \bigg| f \in \Omega_f \right\} 
    \geq 1-2\mathcal{T}n_r\e^{-\frac{\epsilon^2}{2\eta^{\star}_1\lambda^2_{\omega}}}.
\end{align}
Additionally, the following inequalities hold
\begin{align*}      
  \sum^{k_1+\mathcal{T}}_{k=k_1}\left\|\E[r(k)]\right\|_{\infty} 
    \geq \frac{1}{\sqrt{n_r}}\sum^{k_1+\mathcal{T}}_{k=k_1} \left\|\E[r(k)]\right\|_{2} 
    = \frac{1}{\sqrt{n_r}} \sum^{k_1+\mathcal{T}}_{k=k_1} \left\|\T_{fr} [f(k)]\right\|_{2}
    \geq \frac{\sqrt{\eta^{\star}_2}}{\sqrt{n_r}} \mathcal{T}\underline{f},
\end{align*}
where the first inequality is derived from the equivalence between vector norms and the second inequality follows from the result in Theorem~\ref{thm:FD filter design}, i.e.,~$\|\T_{fr}\|^2_{\Hmin(\Theta_m)} \geq \eta^{\star}_2$, and~$\|f(k)\|_2 \geq \underline{f}$ for~$f \in \Omega_f$.
To make sure that $\epsilon$ is positive, let 
\begin{align*}
    \epsilon = \frac{1}{\mathcal{T}}\sum^{k_1+\mathcal{T}}_{k=k_1}\|\E[r(k)]\|_{\infty} - J_{th} > \underline{f} \sqrt{\eta^{\star}_2/n_r} - J_{th}  > 0.
\end{align*}
Thus, the lower bound of $f$ should satisfy $\underline{f} > J_{th} \sqrt{n_r / \eta^{\star}_2}$.
Finally, from inequalities~\eqref{eq:FDR ineq1}, we obtain
\begin{align*}
    \mPr\left\{J(r) > J_{th} \big| f\in\Omega_f \right\} 
    &= \mPr \left\{\frac{1}{\mathcal{T}}  \sum^{k_1+\mathcal{T}}_{k=k_1} \|r(k)\|_2 > J_{th} \bigg| f\in\Omega_f \right\} \\
    &\geq \mPr\left\{\frac{1}{\mathcal{T}}\sum^{k_1+\mathcal{T}}_{k=k_1} \|r(k)\|_{\infty} > J_{th} \big| f\in\Omega_f \right\} \\
    &\geq 1-2\mathcal{T}n_r\e^{-\frac{\epsilon^2}{2\eta^{\star}_1\lambda^2_{\omega}}}
    \geq 1-2\mathcal{T}n_r\e^{-\frac{\left(\underline{f}\sqrt{\eta^{\star}_2/n_r} - J_{th}\right)^2}{2\eta^{\star}_1\lambda^2_{\omega}}} .
\end{align*}
This completes the proof.
\end{proof}

\subsection{Proofs of results in fault estimation}\label{subsec:FE proof}
To prove Theorem~\ref{Thm:FE filter design}, the covariance matrix of the output of an LTI system driven by white noise is computed through the following lemma.
\begin{Lem}[Covariance of the residual]\label{lem:cov fest}
Consider the expression of the residual in~\eqref{eq:residual format} with the unknown signal~$X$ decoupled. The noise~$\omega$ is assumed to be i.i.d. white noise and the fault~$f$ is considered to be deterministic. The covariance matrix of~$r$ is given by
\begin{align*}
    \E\left[(r(k)-\E[r(k)])(r(k)-\E[r(k)])^{*}\right] 
    = \frac{1}{2\pi} \int^{\pi}_{-\pi} \T_{\omega r}(\e^{j\theta}) \E\left[\omega(k) \omega^{*}(k) \right] \T^{*}_{\omega r}(\e^{j\theta}) \diff \theta.
\end{align*}
\end{Lem}

\begin{proof}
Let~$h_{\omega r}(k)$ be the impulse response of~$\T_{\omega r}(\mfq)$. The covariance function of~$r(k)$ denoted by~$\mathcal{V}_{r}(\tau)$ for~$\tau \in \N$ can be written as
\begin{align*}
    \mathcal{V}_{r}(\tau) &= \E\left[(r(k+\tau)-\E[r(k+\tau)])(r(k)-\E[r(k)])^{*}\right] \\
    &=\E\left[\left(\sum^{\infty}_{m=0} h_{\omega r}(m) \omega(k+\tau-m)\right) \left(\sum^{\infty}_{l=0}h_{\omega r}(l)\omega(k-l)\right)^{*}\right] \\
    &=\sum^{\infty}_{m=0}\sum^{\infty}_{l=0} h_{\omega r}(m) \E\left[\omega(k+\tau-m) \omega^{*}(k-l) \right] h^{*}_{\omega r}(l)\\
    &=\sum^{\infty}_{m=0}\sum^{\infty}_{l=0} h_{\omega r}(m) \mathcal{V}_{\omega}(\tau-m+l) h^{*}_{\omega r}(l),
\end{align*}
where~$\mathcal{V}_{\omega}(\tau-m+l)$ is the covariance function of $\omega$. 
By applying the $Z$-transform on $\mathcal{V}_{r}(\tau)$, the spectrum of~$r(k)$ denoted by~$\Gamma_{r}(\mfq)$ is derived as
\begin{align*}
    \Gamma_{r}(\mfq) &= \sum^{\infty}_{k=-\infty} \mathcal{V}_{r}(k) \mfq^{-k} \\
    &= \sum^{\infty}_{k=-\infty} \sum^{\infty}_{m=0}\sum^{\infty}_{l=0} h_{\omega r}(m) \mathcal{V}_{\omega}(k-m+l) h^{*}_{\omega r}(l) \mfq^{-(k-m+l)}\mfq^{-m}\mfq^{l} \\
    &=\sum^{\infty}_{m=0}h_{\omega r}(m)\mfq^{-m} \sum^{\infty}_{k=-\infty}\mathcal{V}_{\omega}(k-m+l)\mfq^{-(k-m+l)} \sum^{\infty}_{l=0} h^{*}_{\omega r}(l) \mfq^{l} \\
    &=\T_{\omega r}(\mfq) \Gamma_{\omega}(\mfq) \T^*_{\omega r}(\mfq^{-*}), 
\end{align*}
where $\Gamma_{\omega}(\mfq)$ is the spectrum of $\omega$. 
When $\tau=0$, since $\omega$ is an uncorrelated sequence, we have
\begin{align*}
    \mathcal{V}_{r}(0) &=\E\left[(r(k)-\E[r(k)])(r(k)-\E[r(k)])^{*}\right]\\
    &=\sum^{\infty}_{m=0}\sum^{\infty}_{l=0} h_{\omega r}(m) \E\left[\omega(k-m) \omega^{*}(k-l) \right] h^{*}_{\omega r}(l)\\
    &=\sum^{\infty}_{m=0} h_{\omega r}(m) \E\left[\omega(k) \omega^{*}(k) \right] h^{*}_{\omega r}(m) \\
    &=\frac{1}{2\pi j} \int^{\pi}_{-\pi} \Gamma_{r}(\mfq)\mfq^{-1} \diff \mfq \\
    &= \frac{1}{2\pi j} \int^{\pi}_{-\pi} \T_{\omega r}(\mfq) \E\left[\omega(k) \omega^{*}(k) \right] \T^*_{\omega r}(\mfq) \mfq^{-1} \diff \mfq,
\end{align*}
where the inverse $Z$-transform and the fact that $\mfq^{-*} = \mfq$ on the unit circle are used in the last two equations. Also, due to the derivative~$\diff \mfq /\diff \theta = j\e^{j\theta}$, it holds that
\begin{align*}
    \E\left[(r(k)-\E[r(k)])(r(k)-\E[r(k)])^{*}\right] 
    = \frac{1}{2\pi} \int^{\pi}_{-\pi} \T_{\omega r}(\e^{j\theta}) \E\left[\omega(k) \omega^{*}(k) \right] \T^*_{\omega r}(\e^{j\theta}) \diff \theta.
\end{align*}
This completes the proof.
\end{proof}

\begin{proof}[\textbf{Proof of Theorem~\ref{Thm:FE filter design}}]
First, it is demonstrated in Theorem~\ref{thm:FD filter design} that~\eqref{eq:fest opt1} is equivalent to condition~\eqref{eq:mapping 1}.
Second, to show that~\eqref{eq:fest opt2} implies the satisfaction of~\eqref{eq:mapping 2}, let us recall that~$r = \T_{f r}(\mfq)[f] + \T_{\omega r}(\mfq)[\omega]$, where~$\T_{\omega r}(\mfq) = \bar{\mathcal{N}}\Psi_W(\mfq)$ and~$f$ is assumed to be deterministic. According to Lemma~\ref{lem:cov fest}, the covariance of~$r$ satisfies
\begin{align}\label{eq:var fest}
    &\E\left[(r(k)-\E[r(k)])(r(k)-\E[r(k)])^{*}\right] \notag\\
    = &\frac{1}{2\pi} \int^{\pi}_{-\pi} \T_{\omega r}(\e^{j\theta}) \E[\omega(k) \omega^{*}(k)] \T^*_{\omega r}(\e^{j\theta}) \diff \theta \notag\\
    \preceq &\frac{\lambda^2_{\omega}}{2\pi} \int^{\pi}_{-\pi} \T_{\omega r}(\e^{j\theta})  \T^*_{\omega r}(\e^{j\theta}) \diff \theta \notag\\
    = &\bar{\mathcal{N}} \frac{\lambda^2_{\omega}}{2\pi}\int^{\pi}_{-\pi} \Psi_{W}(\e^{j\theta})  \Psi_{W}^{*}(\e^{j\theta}) \diff \theta \bar{\mathcal{N}}^{\top} = \lambda^2_{\omega} \bar{\mathcal{N}} \Phi \bar{\mathcal{N}}^{\top},
\end{align}
where the inequality holds due to its demonstration through Taylor series expansion and comparison of terms of the same power for~$\phi$ (defined in Lemma~\ref{lem:sub-Gaussian con}). 
It can be shown that for sub-Gaussian random variables,~$ \E[\omega(k) \omega^{*}(k)] \preceq \lambda^2_{\omega}I$. 
As a result, condition~\eqref{eq:mapping 2} which is introduced to suppress the effect of the noise on~$r$ can be achieved by bounding the trace of~$\bar{\mathcal{N}} \Phi \bar{\mathcal{N}}^{\top}$. This also coincides with the $\mathcal{H}_2$ norm.

The last part of the proof shows that the relaxed condition~\eqref{eq: relaxed FE con} can be realized through~\eqref{eq:fest opt3}. 
Note that the singular values of a complex matrix $M_C = M_R + jM_I$ are equal to those of the augmented matrix $\begin{bmatrix} M_R &-M_I \\M_I &M_R \end{bmatrix}$ derived from~$M_C$.
Therefore, constraining the $2$-norm of the augmented matrix in~\eqref{eq:fest opt3}, which is constructed using the real and imaginary parts of $T_{fr}(\e^{j\theta_i})-I$, i.e., $\bar{\mathcal{N}}\mathcal{R}_i-I$ and $\bar{\mathcal{N}}\mathcal{I}_i$, is equivalent to constraining $\left\| \T_{fr}(\e^{j\theta_i})-I \right\|^2_{2}$. 
This completes the proof.
\end{proof}

\begin{proof}[\textbf{Proof of Corollary~\ref{cor:Analytical solution}}]
The Lagrange function of~\eqref{eq:fest opt} is
\begin{align*}
    \mathcal{L} (\bar{\mathcal{N}},\gamma) 
    =  \beta \textup{Trace}\left[\bar{\mathcal{N}} \Phi \bar{\mathcal{N}}^{\top}\right] + \sum^{(d_N+2)(n_x+n_d)}_{i=1} \gamma^{\top}_i \bar{\mathcal{N}} \bar{H}_i  +\frac{1-\beta}{\kappa} \sum^{\kappa}_{i=1} \left\|\begin{bmatrix}
        \bar{\mathcal{N}}\mathcal{R}_i - I &-\bar{\mathcal{N}}\mathcal{I}_i\\
        \bar{\mathcal{N}}\mathcal{I}_i &\bar{\mathcal{N}}\mathcal{R}_i - I
    \end{bmatrix}\right\|^2_F, 
\end{align*}
where $\gamma = [\gamma_1 ~\dots ~\gamma_{(d_N+2)(n_x+n_d)}]$ with $\gamma_i \in \R^{n_f}$ is the Lagrange multiplier. $\bar{H}_i$ is the $i$-th column of $\bar{H}$. According to the definition of Frobenius norm 
\begin{align*}
     \left\|\begin{bmatrix}
        \bar{\mathcal{N}}\mathcal{R}_i - I &-\bar{\mathcal{N}}\mathcal{I}_i\\
        \bar{\mathcal{N}}\mathcal{I}_i &\bar{\mathcal{N}}\mathcal{R}_i - I
    \end{bmatrix}\right\|^2_F = 2\textup{Trace}\left[(\bar{\mathcal{N}}\mathcal{R}_i - I)(\bar{\mathcal{N}}\mathcal{R}_i - I)^{\top}+\bar{\mathcal{N}} \mathcal{I}_i\mathcal{I}^{\top}_i\bar{\mathcal{N}}^{\top}\right].
\end{align*}
Taking the partial derivative of~$\mathcal{L} (\bar{\mathcal{N}},\gamma)$ yields
\begin{align*}
     \frac{\partial \mathcal{L}(\bar{\mathcal{N}},\gamma)}{\partial \bar{\mathcal{N}}} 
     = 2\beta \bar{\mathcal{N}} \Phi +
      \frac{4(1-\beta)}{\kappa}\sum^{\kappa}_{i=1}\left(\bar{\mathcal{N}} \mathcal{R}_i\mathcal{R}^{\top}_i - \mathcal{R}^{\top}_i + \bar{\mathcal{N}} \mathcal{I}_i\mathcal{I}^{\top}_i\right) 
      + \sum^{(d_N+2)(n_x+n_d)}_{i=1} \gamma_i \bar{H}^{\top}_i.
\end{align*}
Then, setting the partial derivative to zero and considering the equality constraint~\eqref{eq:fest opt1} leads to
\begin{align*}
    &\begin{bmatrix}
        \bar{\mathcal{N}} &\gamma
    \end{bmatrix}
    \begin{bmatrix}
        2\beta  \Phi +\frac{4(1-\beta)}{\kappa}\sum^{\kappa}_{i=1}\left( \mathcal{R}_i\mathcal{R}^{\top}_i  +  \mathcal{I}_i\mathcal{I}^{\top}_i\right) &\bar{H} \\
        \bar{H}^{\top} &0
    \end{bmatrix}  = \begin{bmatrix}
        \frac{4(1-\beta)}{\kappa}\sum^{\kappa}_{i=1} \mathcal{R}^{\top}_i  &0
    \end{bmatrix}.
\end{align*} Solving this equation provides the analytical solution. 
This completes the proof.
\end{proof}

\begin{proof}[\textbf{Proof of Proposition~\ref{prop: opt gap}}]
   Let us first show that the upper bound holds. 
    Since the optimization problem~\eqref{eq:FE opt} is an exact reformulation of Problem 2, applying the AO approach to solve~\eqref{eq:FE opt} leads to the convergence of the objective function value to the optimal value~$\mathcal{J}^{\star}$ of Problem~2. 
    Thus, the derived objective function value, i.e.,~$\beta\eta^{\star}_{1,AO} + (1-\beta) \eta^{\star}_{3,AO}$, is an upper bound on~$\mathcal{J}^{\star}$.

    In the second part of the proof, the satisfaction of the lower bound is proved by contradiction. Suppose that 
    \begin{align*}
        \min_{\mathcal{N}(\mfq)} \max_{\theta_i} \| \T_{fr}(\e^{j\theta_i}) - I \|^2_2 \geq \min_{\mathcal{N}(\mfq)}  \|\T_{fr}(\e^{j\theta}) - I\|^2_{\mathcal{H}_{\infty}(\Theta_m)}, ~\forall \Theta_m \in \bar{\Theta}.
    \end{align*}
   Let~$\mathcal{N}^{\star}(\mfq)$ and $\mathcal{N}^{\star}_{RR} (\mfq)$ denote the optimal solutions to
   \begin{align*}
      \min_{\mathcal{N}(\mfq)}  \|\T_{fr}(\e^{j\theta}) - I\|^2_{\mathcal{H}_{\infty}(\bar{\Theta})} ~\text{and} ~\min_{\mathcal{N}(\mfq)} \max_{\theta_i}  \| \T_{fr}(\e^{j\theta_i}) - I \|^2_2, 
   \end{align*}
   respectively. 
    Recall the definition of the restricted $\mathcal{H}_{\infty}$ norm. For all sampling frequency points $\theta_i$, it holds that 
    \begin{align*}
        \max_{\theta_i} \| \T_{fr}(\e^{j\theta_i}, \mathcal{N}^{\star}_{RR} (\mfq)) - I \|^2_2 
        &\geq \sup_{\theta \in \bar{\Theta}} \|\T_{fr}(\e^{j\theta}, \mathcal{N}^{\star} (\mfq)) - I \|^2_2  \\
        &\geq \| \T_{fr}(\e^{j\theta_i},\mathcal{N}^{\star}(\mfq)) - I \|^2_2,
    \end{align*} 
    which contradicts the fact that $\mathcal{N}^{\star}_{RR} (\mfq)$ is the optimal solution to ~$\min_{\mathcal{N}(\mfq)} \max_{\theta_i} \| \T_{fr}(\e^{j\theta_i}) - I \|^2_2$. 
    Thus, we have $\min_{\mathcal{N}(\mfq)} \max_{\theta_i} \| \T_{fr}(\e^{j\theta_i}) - I \|^2_2 \leq \min_{\mathcal{N}(\mfq)}  \|\T_{fr}(\e^{j\theta}) - I\|^2_{\mathcal{H}_{\infty}(\bar{\Theta})}$. 
    Additionally, the constraints~\eqref{eq:mapping 1} and \eqref{eq:mapping 2} on noise suppression and disturbance decoupling are identical in both Problem~2 and Problem~2r. As a result, the optimal objective value of Problem~2r, obtained by solving~\eqref{eq:fest opt}, serves as a lower bound for~$\mathcal{J}^{\star}$. This completes the proof. 
\end{proof}

%%%%%%%%%%%%%%%%%%%%%%%%%%%%%%%%%%%%%%%%%%%%%%%%%%%%%%%%%%%%%%%%
%%%%%%%%%%%%%%%%%%%%%%%%%%%%%%%%%%%%%%%%%%%%%%%%%%%%%%%%%%%%%%%%
%                       Simulation results
%%%%%%%%%%%%%%%%%%%%%%%%%%%%%%%%%%%%%%%%%%%%%%%%%%%%%%%%%%%%%%%%
\section{Simulation results} \label{sec:simulation}
The effectiveness of the proposed FDE methods is validated on a non-minimum phase hydraulic turbine system and on a multi-area power system.

\subsection{A hydraulic turbine system}
Note that non-minimum phase systems are prevalent in a wide range of practical applications, such as aerospace engineering, power systems, etc. 
The ubiquity of non-minimum phase systems in the real-world underscores the critical importance of developing fault diagnosis methods for them. 
However, the inherent characteristics of non-minimum phase systems, particularly their unstable inverse response behavior, pose significant challenges in fault estimation, as discussed in Remark~\ref{rem: non minimum}.
To address this issue, we develop fault estimation filter design techniques that focus on specific frequency bands of interest, offering significant advantages in estimation performance compared to existing results. 
To verify the performance, a hydraulic turbine system from~\cite{wang2015disturbance} is considered as follows 
\begin{align*}
    y = \frac{-0.183s+1.4}{0.2136s^3+2.445s^2+5.911s+0.45}(u+f_u),
\end{align*}
where $u$ and $y$ are the turbine valve and the turbine speed, respectively. 
The fault on the turbine valve is denoted as $f_u$.
The system has an unstable zero at $7.65$. 
To facilitate diagnosis filter design, the transfer function of the hydraulic turbine system is transferred into the state-space representation and discretized with the sampling period $0.1$s. 
In addition, though modeling errors exist caused by discretization, their effects are negligible when the sampling interval is sufficiently small.

In this part, methods developed in Theorem~\ref{Thm: exact FE} (ER, exact reformulation) and Theorem~\ref{Thm:FE filter design} (RR, relaxed reformulation) are used to estimate the fault signal in the absence of disturbances and noise. 
In the simulation, the proposed estimation methods are compared with the UIO (unknown input observer) method~\cite{gao2015fault}, the LS (least square) method~\cite{wan2016data}, and the IUIE (\mbox{inversion-based} unknown input estimation) method~\cite{wan2017fault}. Both the UIO, LS, and IUIE methods are proven to be asymptotically unbiased estimation methods under certain conditions. 

The frequency range of interest is~$\Theta=[0,0.2]$ and the fault signal is~$f(k) = 0.05\sin(0.1k)+0.06\sin(0.15k)$ sampled from the corresponding continuous-time signal with the sampling time $0.1$s here.
First, a stable denominator is selected as~$a(\mfq) = (\mfq-0.1)^5$ and $6$ frequency points are chosen when using the RR method in Theorem~\ref{Thm:FE filter design} to design the fault estimation filter. 
By solving the optimization problem~\eqref{eq:fest opt}, the numerator~$\mathcal{N}^{\star}_{RR}(\mfq)$ and the optimal value $\bar{\eta}^{\star}_{3,RR} = 0.0534$ are obtained.
Then, the denominator~$a(\mfq)$ is fixed and~$\mathcal{N}^{\star}_{RR}(\mfq)$ is used as the initial condition to design the fault estimation filter when using the ER method in Theorem~\ref{Thm: exact FE} and Algorithm~\ref{alg:algorithm_1}. 
The obtained value of the objective function is $\eta^{\star}_{3,AO} = 0.0.0764$ after~$5$ iteration steps.
According to~\eqref{eq: opt gap}, the suboptimality gap is~$0.0534 \leq \mathcal{J}^{\star}  \leq 0.0764$.

Fig.~\ref{fig:NonminPhase fest} presents the fault signal and its estimates obtained by different methods, while errors of fault estimates are illustrated in Fig.~\ref{fig: NonminPhase ef}. 
As illustrated in Fig.~\ref{fig:NonminPhase fest}, both the IUIE and LS methods diverge, while the UIO methods produce high estimation errors. 
In comparison with the above methods, the proposed ER and RR methods offer better estimation performance.
In Fig.~\ref{fig:NonmindiffdN ef}, it is further demonstrated that increasing the degree of the RR filter can reduce the estimation error.

\begin{figure}[h]
    \centering
    \begin{minipage}{0.5\textwidth}
    \centering
    \includegraphics[scale=0.7]{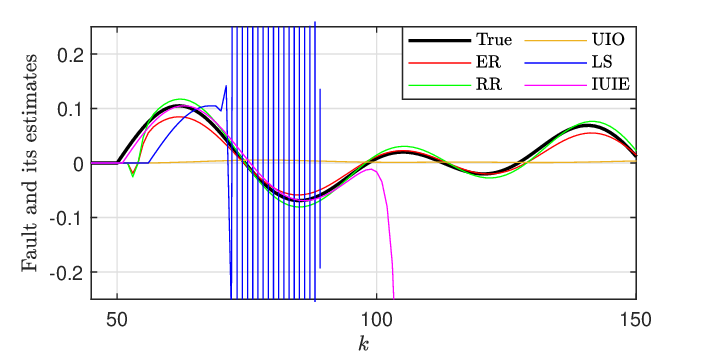} 
    \caption{\small Fault and its estimates generated using different methods.}\label{fig:NonminPhase fest}
    \end{minipage}
    \begin{minipage}{0.48\textwidth}
    \centering
    \includegraphics[scale=0.69]{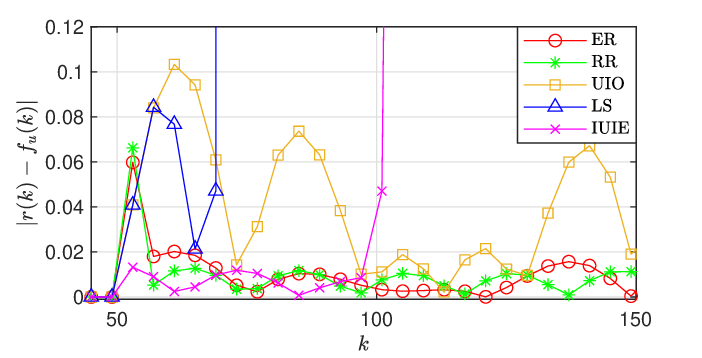} 
    \caption{\small Errors of fault estimates.}
    \label{fig: NonminPhase ef}
    \end{minipage}
\end{figure}

\begin{figure}[h]
    \centering
    \includegraphics[scale=0.7]{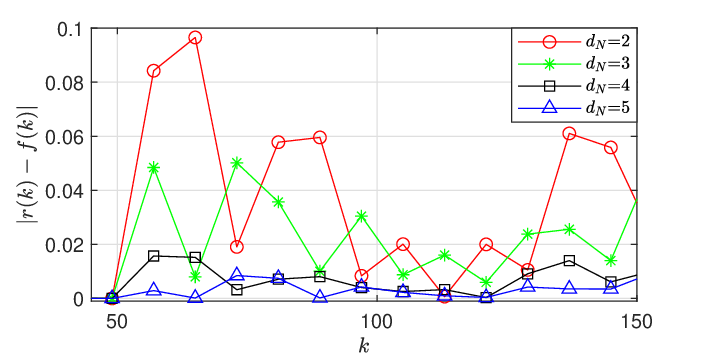} 
    \caption{\small Errors of fault estimates with different degrees.}\label{fig:NonmindiffdN ef}
\end{figure}

\subsection{Multi-area power systems}
\subsubsection{System description}
Consider a multi-area power system described in~\cite{pan2019static}. Suppose each area of the power system can be represented by a model with equivalent governors, turbines, and generators. Then, in area~$i$ for~$i \in \{1,2,3\}$, the dynamics of frequency~$\Delta \bdw_i$ can be written as 
\begin{align}\label{eq:power sys}
    \left\{ \begin{array}{ll}
         \Delta \dot{\bdw}_i &= \frac{\bdw_0}{2 h_i S_{B_i}} (\Delta p_{m_i} - \Delta p_{tie_i} - \Delta p_{d_i} - \frac{1}{D_{li}} \Delta \bdw_i) ,\\
         \Delta p_{m_i} &= \sum^{Gen_{i}}_{g=1} \Delta p_{m_{ig}}, 
        ~\Delta p_{tie_i} = \sum_{j \in N_{br_i}} \Delta p_{tie_{ij}}, \\
         \Delta \dot{p}_{m_{ig}} &= -\frac{1}{T_{ch_{ig}}} (\Delta p_{m_{ig}} + \frac{1}{S_{i}} \Delta \bdw_i - \rho_{ig} \Delta p_{agc_i}), \\
         \Delta \dot{p}_{tie_{ij}} &= 2\pi P_{T_{ij}} (\Delta \bdw_i - \Delta \bdw_j), \\
         ACE_i &= \zeta_i \Delta \bdw_i + \Delta p_{tie_i}, \\
         \Delta \dot{p}_{agc_i} &= -K_{I_i} ACE_i,
    \end{array}
    \right.
\end{align}
where~$h_i$ represents the equivalent inertia constant, $\bdw_0$ denotes the nominal frequency, $S_{B_i}$ is the power base, $\Delta p_{m_i}$ denotes the total generated power, $\Delta p_{tie_i}$ denotes the total tie-line power exchanges from area~$i$, $\Delta p_{d_i}$ denotes the deviation caused by the load, and~$1/D_{li} \Delta \bdw_i$ is the deviation caused by the frequency dependency of the load. 
Let~$G_{en_i}$ and~$N_{br_i}$ be the number of generators and the set of areas that connect to area~$i$, respectively. The term~$\Delta p_{m_{ig}}$ denotes the power generated by the $g$th generator,~$\Delta p_{tie_{ij}}$ is the power exchanges between area~$i$ and~$j$, and~$P_{T_{ij}}$ is the maximum transfer power on the line, which is assumed to be constant. It holds that~$\Delta p_{tie_{ij}} = -\Delta p_{tie_{ji}}$.
For the dynamics of~$\Delta p_{m_{ig}}$, $T_{ch_{ig}}$ is the governor-turbine's time constant, and $S_{i}$ is the drop coefficient. 
The term~$\Delta p_{agc_i}$ is the automatic generation control (AGC) signal and $\rho_{ig}$ is the participating factor, i.e.,~$\sum^{Gen_i}_{g=1} \rho_{ig}=1$.
The area control error signal is denoted by~$ACE_i$ and~$\zeta_i$ is the frequency bias factor. The AGC signal~$\Delta p_{agc_i}$ in the last line of~\eqref{eq:power sys} is in integration of~$ACE_i$ with the integral gain $K_{I_i}$. The parameters are provided in Table~\ref{tab:init}.

\begin{table}[t]
\begin{center}
\caption{Parameters of the multi-area power system.} \label{tab:init}   
        \begin{tabular}{ cc |cc}
        \hline 
        Name &Values &Name &Values\\
        \hline
        $\bdw_0$  & 60 Hz          &$D_{l1}$ &0.0064 Hz/MW\\
        $h_1$     &4.41 MW/MVA     &$D_{l2}$ &0.0045 Hz/MW \\
        $h_2$     &4.15 MW/MVA     &$D_{l3}$ &0.0056 Hz/MW\\   
        $h_3$     &3.46 MW/MVA     &$Gen_1$  &2\\
        $S_{B_1}$ &1500 MVA        &$Gen_2$  &3\\
        $S_{B_2}$ &2100 MVA        &$Gen_3$  &2\\
        $S_{B_3}$ &1700 MVA        &$\zeta_1$ &500.0064 Hz/MW\\
        $S_1$     &0.002  MW/Hz    &$\zeta_2$ &700.0045 Hz/MW\\
        $S_2$     &0.0014 MW/Hz    &$\zeta_3$ &566.6723 Hz/MW\\
        $S_3$     &0.0018 MW/Hz    &$K_{I_i}$   &0.65  \\
        $\rho_{1i}$,~$\rho_{3i}$ &1/2 &$\rho_{2i}$ &1/3\\
        $P_{T_{12}}$ &2100 MW  &$P_{T_{13}}$ &2100 MW \\
        $P_{T_{23}}$ &2100 MW  &$T_{ch_{ig}}$ &1.4950\\ 
        \hline
        \end{tabular}
  \end{center}
\end{table}

Note that different faults may happen due to the vulnerabilities of multi-area power systems. Here, the following fault scenarios are considered: 
\begin{enumerate}[label=(\roman*), itemsep = 0mm, topsep = 0mm, leftmargin = 8mm]
    \item faults on the tie line between areas that cause deviation in frequency, i.e.,~$\Delta \dot{p}_{tie_{ij}} = 2\pi P_{T_{ij}} (\Delta \bdw_i - \Delta \bdw_j + f_{tie_{ij}})$;
    \item faults on the AGC part of area~$i$, i.e.,~$\Delta \dot{p}_{agc_i} = -K_{I_i} (ACE_i+f_{agc_i})$;
    \item faults on the sensors of area~$i$, i.e.,~$y_i(t) = C_i x_i(t) + D_{f,i} f_{y_i}$, where~$y_i$,~$C_i$ and~$x_i$ are the output, output matrix, and state of area~$i$, respectively. The matrix~$D_{f,i}$ characterizes the sensors that are vulnerable.
\end{enumerate}

Based on the dynamics~\eqref{eq:power sys} and descriptions of the faults, the state-space model of area~$i$ in the presence of faults becomes
\begin{align*}
\left\{ 
    \begin{array}{ll}
     \dot{x}_i (t) &= A_{ii} x_i (t) + B_{d,i} \Delta p_{d_i}(t) + B_{\omega,i}\omega_i(t) + \sum_{j\in N_{br_i}} A_{ij} x_j(t) + B_{f,i} f_i(t)  \\
      y_i(t) &= C_i x_i(t) + D_{\omega,i} \omega_i(t) + D_{f,i} f_{y_i}(t), 
    \end{array}
\right.
\end{align*}
where the state~$x_i =\left[\Delta p_{tie_{i}} ~\Delta \bdw_i ~\{\Delta p_{m_{ig}}\}_{1:Gen_i} ~\Delta p_{agc_i}\right]^{\top}$, ~$f_i = [\{f_{tie_{ij}}\}_{j\in N_{br_i}} ~f_{agc_i}]^{\top}$ is the process fault signal.
Signal~$\omega$ denotes noise in the system. The matrices~$A_{ii}, ~B_{d,i}, ~A_{ij}, ~B_{f,i}, D_{f,i}$ can be obtained based on the dynamics~\eqref{eq:power sys} and the vulnerable parts of area~$i$. The output matrix~$C_i$ is a tall or square matrix with the full column rank, i.g.,~$C_i = I$. The matrices~$B_{\omega,i}$ and~$D_{\omega,i}$ indicate which signal is affected by the noise. 
Stacking the state of each area, i.e.,~$x = [x_1^{\top}, ~x_2^{\top}, ~x_3^{\top}]^{\top}$, and discretizing the system with sampling period~$0.1s$ results in the discrete-time \mbox{state-space} model for the whole three-area power system in the form of~\eqref{eq:SS model}. The system matrices are given by
\begin{align*}
    A &= \begin{bmatrix}
        A_{11} &A_{12} &A_{13}\\
        A_{21} &A_{22} &A_{23}\\
        A_{31} &A_{32} &A_{33}
    \end{bmatrix},
    ~B_{d} = \text{diag}(B_{d,1}, B_{d,2}, B_{d,3}), B = D = 0,\\
    B_{f} &= \text{diag}(B_{f,1}, B_{f,2}, B_{f,3}),
    ~B_{\omega} = \text{diag}(B_{\omega,1}, B_{\omega,2}, B_{\omega,3}),\\
    D_{\omega} &= \text{diag}(D_{\omega,1}, D_{\omega,2}, D_{\omega,3}),
    ~D_{f} = \text{diag}(D_{f,1}, D_{f,2}, D_{f,3}).
\end{align*}

Here, we consider faults in the tie-line of area~$1$, the AGC part of area~$2$, and the measurement of area~$1$. 
The corresponding faulty matrices are
\begin{align*}
    &B_{f,1} = [2\pi P_{T_{12}} ~0 ~0 ~0 ~0]^{\top},~B_{f,2} = [0 ~0 ~0 ~0 ~0 ~-K_{I_2}]^{\top},\\
    &D_{f,1}= [0 ~1 ~0 ~0 ~0]^{\top},~\text{and}~B_{f,3}=D_{f,2}=D_{f,3}=0.
\end{align*}
The unknown loads are~$\Delta p_{d_1}(k) = \Delta p_{d_2}(k)= \Delta p_{d_3}(k) = 1 + v(k)$ with~$v(k)$ denoting the uncertain signal. The signal~$\omega$ is white noise with zero mean and variance~$0.01$. 
The matrices~$B_{\omega}=0$ and~$D_{\omega} = \boldsymbol{1}$, where~$\boldsymbol{1}$ represents a column vector with all elements~$1$.

\subsubsection{Fault detection results}
Suppose that the frequency content of fault signals is $\bar{\Theta} = [0,0.3]$ in the fault detection problem. 
Let us consider process faults first, i.e.,~$f_{tie_{12}}$ and~$f_{agc_2}$, which are zero before~$k=50$ and then become
\begin{align*}
f_{tie_{12}}(k) &=  0.05\sin(0.2k)+0.06\sin(0.3k), ~k>=50 ~\text{and} \\
f_{agc_2}(k) &= 0.08\sin(0.15k)+0.03\sin(0.25k),~k>=50.
\end{align*}

The process of the fault detection task is summarized as: 

\textit{Step 1}. Set the residual dimension and filter degree to $n_r = 3$ and~$d_N = 2$.
Note that the dimension of the filter states is $n_r(d_N+1)=9$, which is smaller than that of the system $n_x=16$.

\textit{Step 2}. Solve the filter coefficients by using the optimization problem in Theorem~\ref{thm:FD filter design} with the AO approach in Algorithm~\ref{alg:algorithm_1}, where the weight~$\alpha = 0.5$. 

\textit{Step 3}. Compute the threshold $J_{th}$ for fault detection based on Theorem~\ref{Thm:performance certificates}, which is $J_{th} = 0.0153$ with the acceptable FAR $\varepsilon_1 = 0.001$ and time interval $\mathcal{T}=10$. 

\textit{Step 4}. Compare the value of the evaluation function~$J(r)$ to $J_{th}$ to render the diagnosis decision.

The fault detection filter developed in the DAE framework is compared with the Luenberger observer designed using fault frequency content information (LO($\Theta$))~\cite{wang2008finite} and the UIO approach designed for the entire frequency range~\cite{gao2015fault}.  
Since the dimensions of residuals generated by LO($\Theta$) and UIO methods are $n_r = n_y = 16$, while $n_r = 3$ in our approach, the evaluation function~$J(r)$ is divided by $n_r$ for comparison, as is the threshold.

\begin{figure}[t]
    \begin{minipage}{0.49\textwidth}
    \centering
    \captionsetup{justification=centering}
    \includegraphics[scale=0.7]{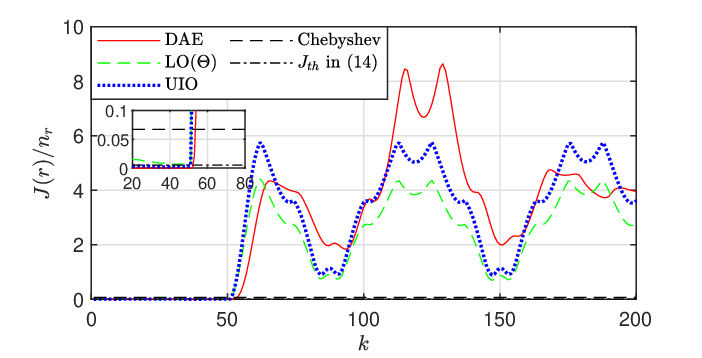} 
    \caption{\small Detection results for $f_{agc_2}$ and $f_{tie_{12}}$.}\label{fig: Fa residual}
    \end{minipage}
    \begin{minipage}{0.49\textwidth}
    \centering
    \captionsetup{justification=centering}
    \includegraphics[scale=0.7]{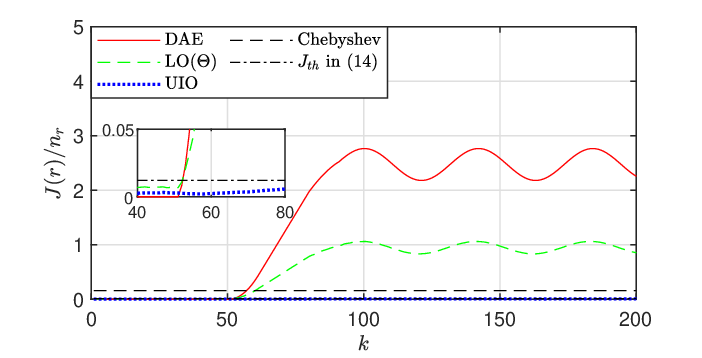} 
    \caption{\small Detection results for~$f_{y_1}$.}
    \label{fig: Fs residual}
    \end{minipage}
\end{figure}

Fig.~\ref{fig: Fa residual} presents the detection results for~$f_{tie_{12}}$ and~$f_{agc_2}$.
One can see that the values of~$J(r(k))/n_r$ remain below the threshold when~$k \leq 50$ and exceed the threshold immediately after faults happen at~$k=50$. Thus, all three approaches have successfully detected the process faults, wherein our proposed method has the best fault sensitivity. 
Moreover, the threshold derived using~\eqref{eq:Jth} is found to be less conservative than the threshold derived using Chebyshev's inequality, i.e.,~$\lambda_{\omega}\sqrt{\mathcal{T} n_r \eta^{\star}_1 /\epsilon_1} = 0.2010$. 

The process of sensor fault detection is the same as above. 
% Differently, the obtained $\Hto$-norm value is~$\|\T_{\omega r}\|_{\Hto} = 0.0159$, the~$\Hmin$ index is~$\| \T_{fr}\|_{\Hmin(\Theta)}=0.0845$ after $5$ iteration steps, and the detection threshold is~$J_{th} = 0.0364$.
The following fault signal is employed to test the detection ability of different methods for sensor faults: 
\begin{align*}
f_{y_1}(k) = \left\{ \begin{array}{l}
     0.005*(k-50), ~80 \geq K > 50,  \\
     0.15+0.02 \sin(0.15k), ~k \geq 80. 
\end{array} \right.
\end{align*}
Fig.~\ref{fig: Fs residual} shows the detection results for~$f_{y_1}$. It can be seen that the UIO approach fails to detect the occurrence of the sensor fault as the amplitude of the fault signal is quite small.
Nonetheless, the LO($\Theta$) method and our proposed method considering the fault frequency information successfully detect the fault. In addition, our method exhibits superior sensitivity to sensor faults compared to the LO($\Theta$) method.

\begin{figure}[t]
    \begin{minipage}{0.49\textwidth}
    \centering
    \captionsetup{justification=centering}
    \includegraphics[scale=0.7]{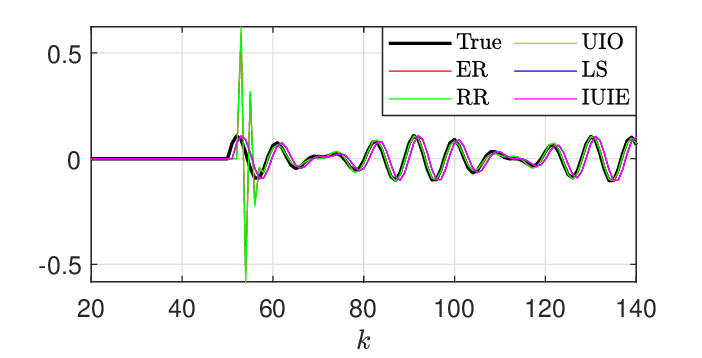} 
    \caption{\small Estimates of $f_{tie{12}}$ without~$\omega$.}
    \label{fig: ftie_now}
    \end{minipage}
    \begin{minipage}{0.49\textwidth}
    \centering
    \captionsetup{justification=centering}
    \includegraphics[scale=0.7]{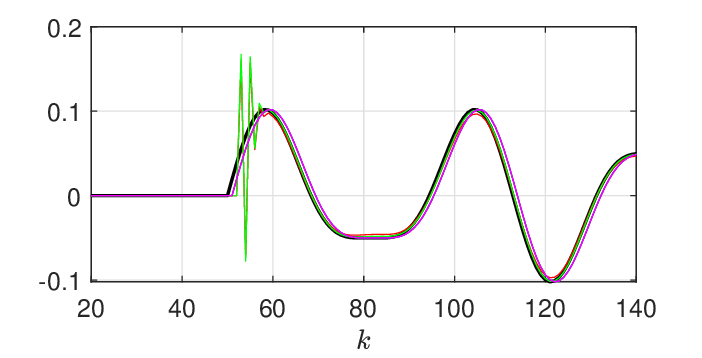} 
    \caption{\small Estimates of $f_{agc_2}$ without~$\omega$.}\label{fig: fagc_now}
    \end{minipage} \\
    \begin{minipage}{0.49\textwidth}
    \centering
    \captionsetup{justification=centering}
    \includegraphics[scale=0.7]{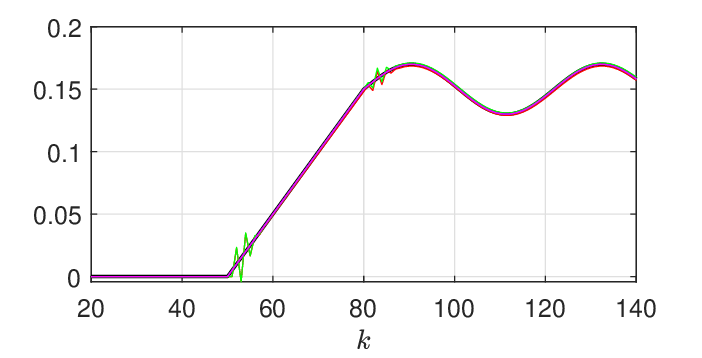} 
    \caption{\small Estimates of $f_{y_1}$ without~$\omega$.}
    \label{fig: fy_now}
    \end{minipage}    
    \begin{minipage}{0.49\textwidth}
    \centering
    \captionsetup{justification=centering}
    \includegraphics[scale=0.7]{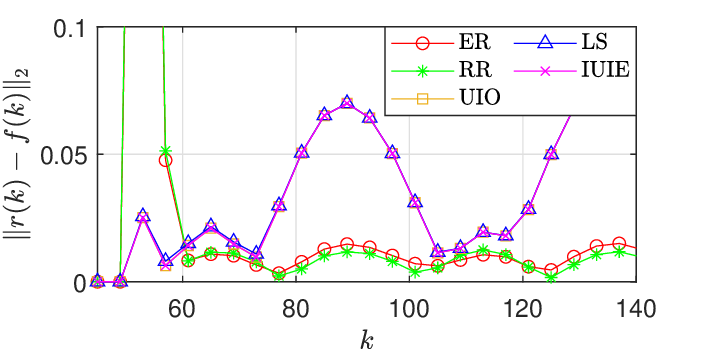} 
    \caption{\small Estimation errors without~$\omega$.}\label{fig: estimation error now}
    \end{minipage}
\end{figure}

\begin{figure}[t]
    \centering
    \includegraphics[scale=0.7]{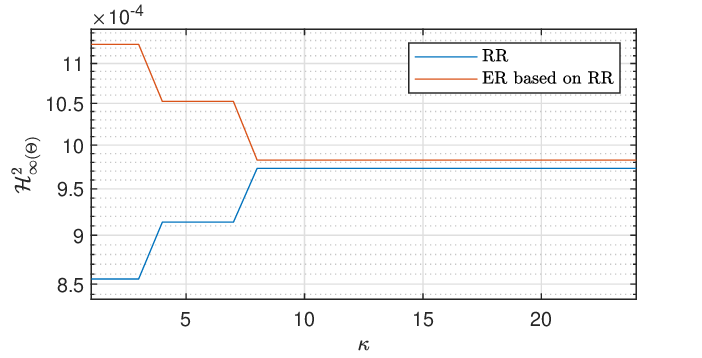}
    \caption{\small Suboptimality gap with different sampling number.}
    \label{fig: subopt gap}
\end{figure}

\begin{figure}[t]
    \begin{minipage}{0.49\textwidth}
    \centering
    \captionsetup{justification=centering}
    \includegraphics[scale=0.7]{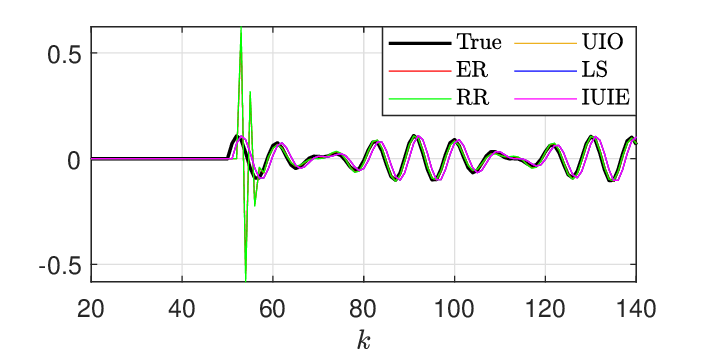} 
    \caption{\small Estimates of $f_{tie_{12}}$ with~$\omega$.}
    \label{fig: ftie_w}
    \end{minipage}
    \begin{minipage}{0.49\textwidth}
    \centering
    \includegraphics[scale=0.7]{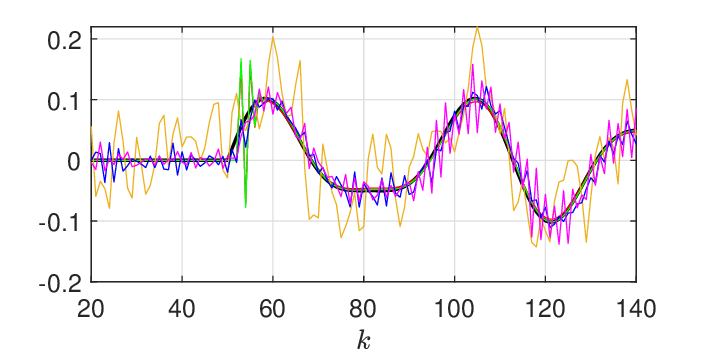} 
    \caption{\small Estimates of $f_{agc_2}$ with~$\omega$.}\label{fig: fagc_w}
    \end{minipage} \\
    \begin{minipage}{0.49\textwidth}
    \centering
    \captionsetup{justification=centering}
    \includegraphics[scale=0.7]{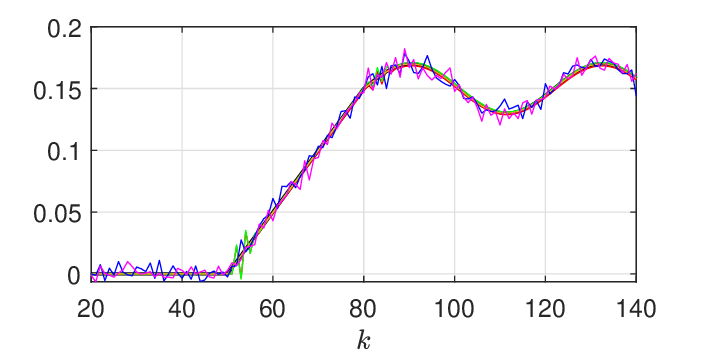} 
    \caption{\small Estimates of $f_{y_1}$ with~$\omega$.\newline}
    \label{fig: fy_w}
    \end{minipage}    
    \begin{minipage}{0.49\textwidth}
    \centering
    \captionsetup{justification=centering}
    \includegraphics[scale=0.7]{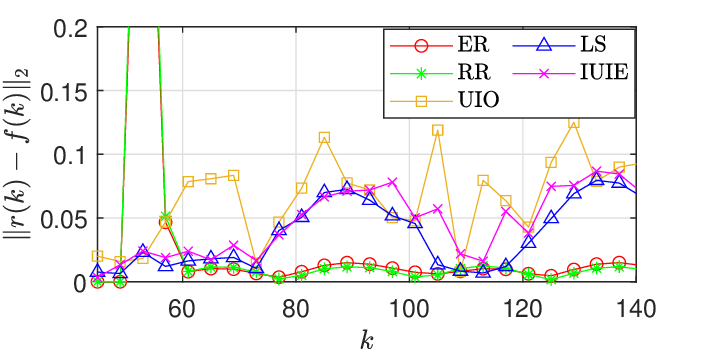}  
    \caption{\small Estimation errors with~$\omega$.}\label{fig: estimation error w}
    \end{minipage}
\end{figure}

\subsubsection{Fault estimation results}
In the fault estimation part, it is supposed that the fault frequency content consists of two disjoint ranges, i.e.,~$\Theta_1 = [0,0.3]$ and~$\Theta_2 = [0.6,0.9]$. The AGC fault signal~$f_{agc_2}$ and the sensor fault signal $f_{y_1}$ remain unchanged with frequencies in~$\Theta_1$. The tie-line fault~$f_{tie_{12}}$ is replaced with 
\begin{align*}
    f_{tie_{12}}(k) &=  0.05\sin(0.8k)+0.06\sin(0.65k), ~k>=50,
\end{align*}
whose frequency is in~$\Theta_2$. The process of the fault estimation task is as follows:

\textit{Step 1}. Set the residual dimension and filter degree to $n_r = n_f = 3$ and~$d_N=4$. 

\textit{Step 2}. Solve two fault estimation filters using the ER method in Theorem~\ref{Thm: exact FE} and the RR method in Theorem~\ref{Thm:FE filter design}, respectively. 
In the ER method, the AO approach is employed to solve~\eqref{eq:FE opt}.
When using the RR method, select a stable denominator~$a(\mfq)$ and some frequency points in~$[0,0.3]$ and~$[0.6,0.9]$ before solving the optimization problem~\eqref{eq:fest opt}.   

\textit{Step 3}. Feeding the control input $u$ and the measurement $y$ into the fault estimation filters yields estimates of fault signals.

To validate the performance of the proposed ER and RR methods, they are compared with the UIO, LS, and IUIE methods in the two cases of no noise and considering noise. 
First, the weight is set to~$\beta = 0$ in the optimization problems~\eqref{eq:FE opt} and~\eqref{eq:fest opt} in the noise-free case.  
The estimation results are presented in Fig.~\ref{fig: ftie_now}-\ref{fig: estimation error now}. 
Specifically, Fig.~\ref{fig: ftie_now}-\ref{fig: fy_now} show the estimates of the tie-line fault~$f_{tie_{12}}$, the AGC fault~$f_{agc_2}$, and the sensor fault~$f_{y_1}$ by different methods. Since the UIO, LS, and IUIE methods both obtain unbiased estimation results with a one-step delay, estimation errors of the three methods are the same as shown in Fig.~\ref{fig: estimation error now}. 
In contrast, the proposed ER and RR methods produce smaller estimation errors than the other three methods. 
Note that though the errors are large at the initial estimation phase, they decrease quickly.
Furthermore, Fig.~\ref{fig: subopt gap} shows the effect of the sampling number of frequency points in the RR method along with the suboptimality gap. For simplicity, a single frequency range~$[0,0.5]$ is considered. The number of frequency points increases from $2$ to $25$, where the new frequency point is added to the previous ones during the process. As a result, the lower bound increases monotonically because more constraints are included in~\eqref{eq:fest opt} when adding frequency points

In the case of considering noise, the weight is set to~$\beta=0.1$.
Since the effect of noise is ignored in the design of the UIO, LS, and IUIE methods, much smaller noise is considered for these three methods.
Fig.~\ref{fig: ftie_w}-\ref{fig: fy_w} depict the estimates of the fault signals in the presence of noise by different methods. 
One can see from Fig.~\ref{fig: fagc_w} that the estimates of the AGC fault signal obtained by the UIO, LS, and IUIE methods are corrupted by noise seriously. 
In contrast, thanks to the noise suppression and design in the specific frequency ranges, the ER and RR methods achieve smaller estimation errors than the other three methods under the effects of noise as illustrated in Fig.~\ref{fig: estimation error w}.   

%%%%%%%%%%%%%%%%%%%%%%%%%%%%%%%%%%%%%%%%%%%%%%%%%%%%%%%%%%%%%%%%%%%
%%%%%%%%%%%%%%%%%%%%%%%%%%%%%%%%%%%%%%%%%%%%%%%%%%%%%%%%%%%%%%%%%%%
%                         Conclusions
%%%%%%%%%%%%%%%%%%%%%%%%%%%%%%%%%%%%%%%%%%%%%%%%%%%%%%%%%%%%%%%%%%%

\section{Conclusions} \label{sec:conslusion}
This paper studies the design methods of FDE filters in the frequency domain for LTI systems with disturbances and stochastic noise.
Based on an integration of residual generation and norm approaches, the optimal design of FDE filters is formulated into a unified optimization framework. 
In future work, a potential research direction is to extend the results to nonlinear systems.

\bibliographystyle{elsarticle-num}
\bibliography{ref}
\end{document}

%% file: style.tex
\makeatletter
\def\@settitle{\begin{center}%
		\baselineskip14\p@\relax
		\normalfont\LARGE\scshape\bfseries
		\@title
	\end{center}%
}

\def\@setauthors{%
  \begingroup
  \def\thanks{\protect\thanks@warning}%
  \trivlist
  \centering\footnotesize \@topsep30\p@\relax
  \advance\@topsep by -\baselineskip
  \item\relax
  \author@andify\authors
  \def\\{\protect\linebreak}%
  \authors%
  \ifx\@empty\contribs
  \else
    ,\penalty-3 \space \@setcontribs
    \@closetoccontribs
  \fi
  \endtrivlist
  \endgroup
}

\makeatother

\makeatletter

\def\subsection{\@startsection{subsection}{2}%
	\z@{.5\linespacing\@plus.7\linespacing}{.5\linespacing}%
	{\normalfont\large\bfseries}}

\makeatother

\makeatother

\usepackage[usenames, dvipsnames]{color}
\definecolor{darkblue}{rgb}{0.0, 0.0, 0.45}

\usepackage[colorlinks	= true,
raiselinks	= true,
linkcolor	= darkblue, %MidnightBlue,
citecolor	= Mahogany,
urlcolor	= ForestGreen,
pdfauthor	= {Peyman Mohajerin Esfahani},
pdftitle	= {},
pdfkeywords	= {},
pdfsubject	= {},
plainpages	= false]{hyperref}

\usepackage{dsfont,amssymb,amsmath,subfigure, graphicx,enumitem,multirow} %,etaremune,savesym}
\usepackage{amsfonts,dsfont,mathtools, mathrsfs,amsthm} 
\usepackage[amssymb, thickqspace]{SIunits}
\usepackage{algorithm,fancyhdr}
\usepackage{algorithmicx}
\usepackage{algpseudocode}

\allowdisplaybreaks
\date{\today}

%% file: commands.tex
\theoremstyle{theorem}
\newtheorem{Thm}{Theorem}[section]
\newtheorem{Prop}[Thm]{Proposition}

\newtheorem{Lem}[Thm]{Lemma}
\newtheorem{Cor}[Thm]{Corollary}
\newtheorem{As}[Thm]{Assumption}

\newtheorem{Rem}[Thm]{Remark}

\newtheorem{Def}[Thm]{Definition}

\newcommand{\mPr}{\mathbf{Pr}}

\newcommand{\bdw}{\boldsymbol{w}}
\newcommand{\mfq}{\mathfrak{q}}

\newcommand{\F}{\mathds{F}}
\newcommand{\T}{\mathds{T}}
\newcommand{\Lto}{\ell_2}
\newcommand{\Tr}{\textup{Trace}}
\newcommand{\Hto}{\mathcal{H}_2}

\newcommand{\Hmin}{\mathcal{H}_{\_}}
\newcommand{\Hinf}{\mathcal{H_{\infty}}}

\newcommand{\R}{\mathbb{R}}

\newcommand{\N}{\mathbb{N}}

\newcommand{\diff}{\mathrm{d}}

\DeclareMathOperator{\e}{e}

\newcommand{\M}{\mathcal{M}}

\newcommand{\E}{\textbf{\textup{E}}}